\documentclass[format=acmsmall, review=false]{acmart}
\usepackage{acm-ec-25}
\usepackage{booktabs} 
\usepackage{cleveref}
\usepackage{tcolorbox}
\usepackage{threeparttable}
\usepackage{tabularx, booktabs}

\newcommand{\bt}{\ensuremath{\mathsf{RTB}}\xspace} 

\newcommand{\re}{\textbf{Symmetry Principle}\xspace}

\newcommand{\fp}{\textbf{FP}\xspace}
\usepackage{macros_pan}

\newcommand{\kapa}{\kappa^*_s}
\newcommand{\kapb}{\kappa^*_m}
\newcommand{\kapc}{\kappa^*_B}
\newcommand{\kapd}{\kappa^*_S}

\newcommand{\bfi}{\mathbf{i}}

\newcommand{\anhai}[1]{{\color{black} #1}}
\makeatletter
\newcommand{\thickhline}{%
    \noalign {\ifnum 0=`}\fi \hrule height 1pt
    \futurelet \reserved@a \@xhline
}
\newcolumntype{"}{@{\hskip\tabcolsep\vrule width 1pt\hskip\tabcolsep}}
\makeatother

\title{Online Matching under KIID: Enhanced Competitive Analysis through Ordinary Differential Equation Systems}
\author{Pan Xu}
\authornote{Department of Computer Science, New Jersey Institute of Technology, Newark, NJ 07102, \texttt{pxu@njit.edu}}  

\begin{abstract}
We consider the (offline) vertex-weighted Online Matching problem under Known Identical and Independent Distributions (KIID) with integral arrival rates. This setting assumes that (1) all edges incident to any offline node carry a uniform weight and (2) every online node has an integer arrival rate. We propose a meta-algorithm, denoted as $\mathsf{RTB}$, featuring \emph{Real-Time Boosting}, where the core idea is as follows. Consider a bipartite graph $G=(I,J,E)$, where $I$ and $J$ represent the sets of offline and online nodes, respectively. Let $\x=(x_{ij}) \in [0,1]^{|E|}$, where $x_{ij}$ for $(i,j) \in E$ represents the probability that edge $(i,j)$ is matched in an offline optimal policy (\emph{a.k.a.} a clairvoyant optimal policy), typically obtained by solving a benchmark linear program (LP). Upon the arrival of an online node $j$ at some time $t \in [0,1]$, $\mathsf{RTB}$ samples a safe (available) neighbor $i \in I_{j,t}$ with probability $x_{ij}/\sum_{i' \in I_{j,t}} x_{i'j}$ and matches it to $j$, where $I_{j,t}$ denotes the set of safe offline neighbors of $j$. As time $t$ progresses, the set $I_{j,t}$ shrinks as more of $j$'s offline neighbors get matched, leading to a boosted sampling distribution for $j$ over time.

In this paper, we showcase the power of Real-Time Boosting by demonstrating that $\mathsf{RTB}$, when fed with $\X^*$, achieves a competitive ratio of $(2e^4 - 8e^2 + 21e - 27) / (2e^4) \approx 0.7341$, where $\X^* \in \{0,1/3,2/3\}^{|E|}$ is a random vector obtained by applying a customized dependent rounding technique due to Brubach \emph{et al.} (Algorithmica, 2020). Our result improves upon the state-of-the-art ratios of 0.7299 by Brubach \emph{et al.} (Algorithmica, 2020) and 0.725 by Jaillet and Lu (Mathematics of Operations Research, 2013). Notably, this improvement does not stem from the algorithm itself but from a new competitive analysis methodology: We introduce an Ordinary Differential Equation (ODE) system-based approach that enables a \emph{holistic} analysis of $\mathsf{RTB}$. We anticipate that utilizing other well-structured vectors from more advanced rounding techniques could potentially yield further improvements in the competitiveness. Additionally, we present an auxiliary algorithm that elucidates the intricate connections between the approaches of Brubach \emph{et al.}, Jaillet and Lu, and our proposed methodology.
\end{abstract}

\begin{document}
\begin{titlepage}
\maketitle
\setcounter{tocdepth}{1} 
\newpage
\tableofcontents

\end{titlepage}

\newpage

\setcounter{page}{1}
\section{Introduction}
The classical model of online matching was first introduced by~\cite{kvv}. Ever since, different online matching models have been proposed and studied; see details in the survey book~\cite{mehta2012online}. Online matching and related models have received considerable attention over the last few decades due to their wide applications in the e-commerce economy, such as ride-sharing services~\cite{xu-aaai-19,teac21}, assortment optimization~\cite{feng2019linear,gong2021online,fata2019multi}, and crowdsourcing markets~\cite{ho2012online,dickerson2018assigning}. In this paper, we consider online matching under the arrival setting of \emph{Known Identical and Independent Distributions} (KIID), which is detailed below.

\subsection{Statement of the Main Model}
Suppose we have a bipartite graph $G = (I,J,E)$, where $I$ and $J$ represent the sets of offline and online nodes, respectively. For each node $i$, let $J_i \subseteq J$ denote the set of neighbors of $i$ in $J$; similarly, we use $I_j \subseteq I$ to denote the set of neighbors of $j$ in $I$. Each edge $(i,j)$ is associated with a positive weight $w_{ij} > 0$. Upon the arrival of an online node $j$, we need to make an immediate and irrevocable decision: either reject $j$ or match $j$ with one of its offline neighbors $i \in I_j$. In the latter case, we gain a weight of $w_{ij}$, and $i$ will be removed permanently. Throughout this paper, we assume, without loss of generality (WLOG), that each offline node $i$ has a unit matching capacity. Our goal is to design an online matching algorithm to maximize the expected total weight. The KIID arrival settings of online nodes can be stated equivalently in the following two ways.

\xhdr{A Discrete Arrival Setting}. Suppose we have a time horizon of $T$ rounds. During each round $k \in [T] := \{1,2,\ldots,T\}$, a single online node $j$ is sampled (referred to as \emph{$j$ arrives}) from $J$ \emph{with replacement} according to a known distribution ${p_j}$ such that $\sum_{j \in J} p_j = 1$. Note that the sampling process is independent and identical across the $T$ online rounds. For each $j$, let $r_j = T \cdot p_j$, which is called the \emph{arrival rate} of $j$, with $\sum_{j \in J} r_j = T$.

\xhdr{A Continuous Arrival Setting}. Suppose we have a time horizon scaled to $[0,1]$. At any time $t \in [0,1]$, each online node $j$ arrives according to an \emph{independent} Poisson process with a \emph{homogeneous} rate of $r_j$.

According to the work by~\cite{huang2021online}, the above two arrival settings are equivalent for competitive analysis as $T$ approaches infinity. This paper specifically focuses on the problem of \emph{(offline) vertex-weighted online matching under KIID with integral arrival rates}. This implies the following: (1) All edges incident to any offline node $i$ carry a uniform weight, denoted as $w_{ij} = w_i$ for all $j \in J_i$; (2) All arrival rates of online nodes assume integer values. Based on (2), we can further assume WLOG that every online node has a unit arrival rate, meaning $r_j = 1$ for all $j \in J$, achieved by creating $r_j$ copies for each online node $j$.\footnote{This is a common practice in studying online matching under KIID with integral arrival rates; see, \eg~\cite{bib:Jaillet,brubach2020online}.}

Throughout this paper, we will interchangeably refer to the above two arrival settings to facilitate competitive analysis. Specifically, \emph{we will use $k \in [T]$ to index discrete rounds and $t \in [0,1]$ to index continuous time}.

\subsection{A Meta Algorithm and Main Contributions}
The idea of real-time boosting is quite simple. Consider a given fractional vector $\x \in [0,1]^{|E|}$, where $x_{ij}$ for $(i,j) \in E$ represents the probability that edge $(i,j)$ gets matched in any offline optimal policy (\emph{a.k.a.} a clairvoyant optimal policy). The vector $\x$ can typically be obtained by solving a certain benchmark linear program (LP). Consider a given online node $j$, and let $I_{j,t} \subseteq I_{j}$ be the set of safe (unmatched) neighbors of $j$ at time $t \in [0,1]$. Observe that the set $I_{j,t}$ keeps shrinking over time as more and more of $j$'s offline neighbors become unavailable. Thus, when $j$ arrives at time $t$, a natural strategy to boost the sampling distribution for $j$ is to scale up the sampling probabilities for each \emph{safe} offline neighbor $i \in I_{j,t}$ proportional to its original mass of $x_{ij}$. Specifically, we ignore all matched offline neighbors $\tilde{i} \in I_j - I_{j,t}$ and sample one safe neighbor $i \in I_{j,t}$ with probability $x_{ij}/\sum_{i' \in I_{j,t}} x_{i',j}$. We offer a formal presentation of a meta algorithm featuring real-time boosting in Algorithm~\ref{alg:meta}.

\begin{algorithm}[ht!]
\caption{A meta algorithm featuring real-time boosting (\bt) parameterized by $\x \in [0,1]^{|E|}$.}\label{alg:meta}
\DontPrintSemicolon
\tcc{\bluee{The parameter vector $\x \in [0,1]^{|E|}$ lies in the matching polytope of the input graph $G$, which is assumed to be feasible for some benchmark linear program (LP).}}
When an online node of type $j \in J$ arrives at time $t \in [0,1]$:\;
Let $I_{j,t}$ be the set of safe (unmatched) neighbors of $j$ at time $t$.\;
Sample a safe neighbor $i \in I_{j,t}$ with probability $x_{ij}/\sum_{i' \in I_{j,t}} x_{i',j}$ and match it with $j$. \label{alg:meta/s2}
\end{algorithm}

\begin{lemma}[Appendix~\ref{app:well-bt}]\label{lem:well-bt}
The meta algorithm $\bt$ (Algorithm~\ref{alg:meta}) achieves a competitive ratio of no more than $1 - 1/\sfe$ for vertex-weighted online matching under KIID with integral rates, even when it is fed with an optimal solution of the natural LP in~\cite{huang2021online}.
\end{lemma}

We defer the proof of Lemma~\ref{lem:well-bt} to Appendix~\ref{app:well-bt}. Lemma~\ref{lem:well-bt} underscores the importance of supplying $\bt$ with a well-structured vector $\x = (x_e)$, where each $x_e$ is constrained to a limited number of values,\footnote{Throughout this paper, we use $e$ (italic) to represent an edge $e = (i,j) \in E$, while we use $\sfe$ (non-italic) to denote the natural base, approximately equal to 2.718.} as demonstrated in~\cite{bib:Jaillet,brubach2020online}. This approach offers distinct advantages. Firstly, it streamlines the competitive analysis and facilitates the identification of Worst-Scenario (WS) instances where a target algorithm attains its most unfavorable competitive ratio. The rationale is straightforward: a well-structured vector $\x$ considerably restricts the number of potential graph structures. Secondly, as implied by the proof of Lemma~\ref{lem:well-bt} in Appendix~\ref{app:well-bt}, we must modify the sampling distributions of online nodes to balance the performance of offline nodes with different masses so that $\bt$ can surpass the barrier of $1 - 1/\sfe$. A well-structured vector $\x$, once again, simplifies these adjustments in sampling distributions since each online node can possess a rather limited number of non-zero neighbors under $\x$.

\medskip

\emph{In this paper, we showcase the power of $\bt$ by feeding it a well-structured vector, as done in~\cite{brubach2020online}}. Specifically, we supply $\bt$ with a randomized rounded vector $\X \in \{0,1/3,2/3\}^{|E|}$, where $\X$ is obtained by applying dependent rounding~\cite{gandhi2006dependent} to an optimal solution $\y^*$ of the benchmark LP in~\cite{brubach2020online}. For completeness, we restate the LP as follows. 
\smallskip

For each edge $(i,j) \in E$, let $y_{ij} \in [0,1]$ be the probability that $i$ and $j$ get matched in any offline optimal policy (\emph{a.k.a.} a clairvoyant optimal policy). Recall that $w_i > 0$ represents the weight on node $i$, and $I_j \subseteq I$ and $J_i \subseteq J$ denote the set of neighbors of $j$ and $i$ in the input graph $G$, respectively.

\begingroup
\allowdisplaybreaks
\begin{align}\label{lp}
\max & \sum_{i \in I} w_i y_i && \\
& y_i := \sum_{j \in J_i} y_{ij} \le 1 && \forall i \in I; \\
& y_j := \sum_{i \in I_j} y_{ij} \le 1 && \forall j \in J;\\
& 0 \le y_{ij} \le 1 - 1/\sfe && \forall (i,j) \in E; \label{lp-a}\\
& y_{ij} + y_{i,j'} \le 1 - 1/\sfe^2 && \forall i \in I, \forall j, j' \in J_i.\label{lp-b}
\end{align}
\endgroup

Throughout this paper, we refer to the above benchmark LP as \LP-\eqref{lp}. By the work in~\cite{brubach2020online}, the optimal value of \LP-\eqref{lp} serves as a valid upper bound on any offline optimal policy.\footnote{Meanwhile, \LP-\eqref{lp} is a special case of the natural LP proposed in~\cite{huang2021online}.} Observe that Constraints~\eqref{lp-a} and~\eqref{lp-b} play a crucial role in upper bounding the probability of an offline node falling into the WS structure, as shown in Proposition~\ref{pro:main-0}. Our main technical result is stated below.

\begin{theorem}\label{thm:main-1}
Let $\X^* \in \{0,1/3,2/3\}^{|E|}$ be a random vector obtained by applying a specialized version of dependent rounding in~\cite{brubach2020online} to an optimal solution of the benchmark \LP-\eqref{lp}. We claim that $\bt(\X^*)$ with appropriate modifications achieves a competitive ratio of at least $(2 \sfe^4 - 8 \sfe^2 + 21 \sfe - 27) / (2 \sfe^4) \approx 0.7341$ for vertex-weighted online matching under KIID with integral arrival rates, which improves upon the results of $0.7299$ and $0.725$ due to~\cite{brubach2020online} and~\cite{bib:Jaillet}, respectively.
\end{theorem}

\xhdr{Remarks on Theorem~\ref{thm:main-1}}. (i) Let $\y^* = (y_e^*)$ be an optimal solution to \LP-\eqref{lp}. The \tbf{specialized dependent rounding} mentioned in Theorem~\ref{thm:main-1}, denoted by $\mathsf{DR}[\ell]$ with $\ell = 3$ in~\cite{brubach2020online}, refers to the following procedure. First, apply the classical dependent rounding in~\cite{gandhi2006dependent} to the vector $\ell \cdot \y^* = (\ell y_e^*)$, and let $\Y^* = (Y_e^*)$ be the output randomized integer vector. Since each $y_e^* \in [0,1 - 1/\sfe]$ and $\ell y_e^* = 3 y_e^* \le 2$, it follows that $Y_e^* \in \{0,1,2\}$ by the degree-preservation property. Second, set $\X^* = (X_e^*)$ with $X_e^* = Y_e^* / \ell$ such that each $X_e^* \in \{0,1/3,2/3\}$. 

(ii) \emph{Throughout this paper, unless specified otherwise, the mass of an offline/online node refers to that under the rounded vector $\X^*$ instead of $\y^*$ by default}. The appropriate modifications suggested in Theorem~\ref{thm:main-1} are detailed in Section~\ref{sec:mod}. 

(iii) For our model of vertex-weighted online matching under KIID with integral rates, the current best hardness result is $1 - \sfe^{-2} \approx 0.8646$ due to~\cite{bib:Manshadi}, where the upper bound was derived based on an \emph{unweighted} instance. Consequently, this hints at the possibility of attaining a better upper bound specifically for the vertex-weighted scenario studied here.

\medskip

\xhdr{Main Contributions}. 
First, we re-examine and offer insights into the state-of-the-art algorithms for the classical vertex-weighted online matching under KIID. We introduce an auxiliary algorithm in Section~\ref{sec:com}, which unveils the intricate connections between the approaches in~\cite{bib:Jaillet,brubach2020online} and the meta algorithm $\bt$ (Algorithm~\ref{alg:meta}). In Section~\ref{sec:com}, we show that: (1) The algorithms presented in~\cite{bib:Jaillet,brubach2020online} are essentially equivalent to $\bt$ when fed with the same well-structured vectors used in those works; (2) Both studies in~\cite{bib:Jaillet,brubach2020online} opt to analyze a \emph{nuanced version}, whose performance serves as a valid lower bound on that of the exact algorithms presented in the papers. Specifically, the nuanced version assumes that the sampling distribution $\cD_{j,t}$ when online node $j$ arrives at time $t$ gets boosted only when some offline neighbors of $j$ are matched by $j$ itself before $t$. In other words, it ignores potential boosts to $\cD_{j,t}$ resulting from the matching of an offline neighbor $i$ by some online node $j'$ other than $j$, \ie $j' \neq j$ and $j' \in J_i$. Their choice of focusing on the nuanced version instead of the exact algorithms helps circumvent the positive correlation among offline nodes being matched in $\bt$, as elaborated in Section~\ref{sec:chall}, though at the cost of compromising the exact performance.

Second, we propose a holistic competitive analysis directly for the meta algorithm ($\bt$) by harnessing the power of Ordinary Differential Equation (ODE) systems. This approach enables us to develop principles that pinpoint the Worst-Scenario (WS) structures for offline nodes with different configurations under the rounded vector $\X^*$; see Section~\ref{sec:fold} for the definition of the $\re$ and Section~\ref{sec:less} for the \textbf{Worst-Scenario (WS) Principle}. We demonstrate the effectiveness of the new analysis approach by showing that $\bt$ achieves an improved competitive ratio of $(2 \sfe^4 - 8 \sfe^2 + 21 \sfe - 27) / (2 \sfe^4) \approx 0.7341$, as stated in Theorem~\ref{thm:main-1}, when fed with the same well-structured solution as in~\cite{brubach2020online}. As clarified in the previous paragraph, \emph{the improvement in competitiveness is not due to the algorithm itself, but rather a new analysis approach}. We anticipate that utilizing other well-structured vectors from more advanced rounding techniques, \eg by applying $\mathsf{DR}[\ell]$ with some $\ell \ge 4$ to an optimal solution from a benchmark LP, could potentially yield further improvements.

\subsection{A Preliminary Proof of Theorem~\ref{thm:main-1}}\label{sec:preli}

Let $\X \in \{0,1/3,2/3\}^{|E|}$ be the random rounded vector obtained after applying the specialized dependent rounding $\mathsf{DR}[3]$ to an optimal fractional solution of the benchmark \LP-\eqref{lp}. Let $G(\X)$ be the random graph induced by $\X$, where only edges with nonzero mass (under $\X$) are retained. \emph{Throughout this paper, we refer to edges of mass $1/3$ and those of mass $2/3$ in $G(\X)$ as \textbf{small} and \textbf{big}, respectively}. For each offline node $i \in I$, let $X_i := \sum_{j \in J_i} X_{ij} \in \{0, 1/3, 2/3, 1\}$, which is called \emph{the mass of node $i$} in $G(\X)$. Thus, for each node with mass one, it can have only two possible structures in $G(\X)$: either one big edge and one small edge (called \textbf{1B1S}), or three small edges (called \textbf{3S}). The proposition below provides an upper bound on the probability that any offline node has a big edge in the random graph $G(\X)$.

\begin{proposition}[\cite{brubach2020online}]\label{pro:main-0}
The probability that an offline node has a big edge in $G(\X)$ is no more than $2 - 3/\sfe$, where $\X \in \{0,1/3,2/3\}^{|E|}$ is the random rounded vector obtained after applying the specialized rounding ($\mathsf{DR}[3]$) in~\cite{brubach2020online} to an optimal solution of \LP-\eqref{lp}.
\end{proposition}

The result in Proposition~\ref{pro:main-0} is implicitly referenced in~\cite{brubach2020online}. For completeness, we present a proof in Appendix~\ref{app:main-0}, where the proof highlights the indispensability of Constraints~\eqref{lp-a} and~\eqref{lp-b} in the benchmark \LP-\eqref{lp}.\footnote{The necessity of Constraints~\eqref{lp-b} can be seen from the case where an offline node has two edges of nonzero mass before $\mathsf{DR}[3]$, say $1/2$ and $1/2$. After rounding, the offline node has a big edge of $2/3$ with probability one.} Consider a given $\X$ and a given offline node $i$ with a mass of $X_i \in \{0,1/3,2/3,1\}$. The \tbf{Matching Probability per Mass} (MPM) achieved by $\bt(\X)$ is defined as the ratio of the matching probability of $i$ in $\bt$ to $X_i$. By default, we assume any algorithm attains an MPM of $1$ for any offline node $i$ with $X_i = 0$.

Define 
\begin{align}\label{eqn:kap}
\kap^* &= \frac{2 \sfe^4 - 8 \sfe^2 + 21 \sfe - 27}{2 \sfe^4} \approx 0.7341, && \\
\kapa &= \kap^*; && \kapb = 1 + \frac{27}{4} \sfe^{-4} - 12 \sfe^{-3} + 2 \sfe^{-2} \approx 0.7969. \nonumber \\
\kapc &= 1 - 2 \sfe^{-2} \approx 0.7293; && \kapd = 1 - \frac{9}{2} \sfe^{-3} \approx 0.7760. \nonumber 
\end{align}

\emph{Throughout this paper, we refer to $\kap^*$ as the overall target competitiveness to be achieved; $\kapa$ and $\kapb$ as the target MPM for offline nodes of mass $1/3$ and $2/3$, respectively; and $\kapc$ and $\kapd$ as the target MPM for offline nodes of mass one that have one big and one small edge (1B1S) and three small edges (3S), respectively.}\footnote{The subscripts $s$, $m$, $B$, and $S$ represent four scenarios any offline node can fall into in $G(\X)$: small mass ($1/3$), medium mass ($2/3$), big mass one (1B1S), and big mass one (3S), respectively.}

\begin{proposition}\label{pro:main-1}
For any possible realization of $\X$, we claim that: (1) $\bt(\X)$ achieves a \emph{Matching Probability per Mass} (MPM) of at least $\kapc$ and $\kapd$ for any offline nodes with mass one in the forms of one big and one small edge (1B1S) and three small edges (3S), respectively. (2) $\bt(\X)$ achieves an MPM of at least $\kapa$ and $\kapb$ for any offline nodes of mass $1/3$ and $2/3$, respectively.
\end{proposition}

We demonstrate Proposition~\ref{pro:main-1} in Sections~\ref{sec:one},~\ref{sec:less},~\ref{sec:mod}, and~\ref{sec:ws1/3}. In Appendix~\ref{app:main-thm}, we provide a formal proof showing how Propositions~\ref{pro:main-0} and~\ref{pro:main-1} jointly lead to Theorem~\ref{thm:main-1}.

\subsection{Technical Challenges in Competitive Analysis due to Real-Time Boostings}\label{sec:chall}
Real-time boosting in $\bt$ (Algorithm~\ref{alg:meta}) can lead to positive correlations among offline nodes staying safe (available) at any given time, and thus among offline nodes getting matched as well. We emphasize that this positive correlation persists regardless of whether the input vector is well-structured (\eg each entry can take a limited number of discrete values) or not.

\begin{figure}[ht!]
\begin{subfigure}[b]{0.45\textwidth}
 \centering
\begin{tikzpicture}
   \draw (1.5,1) node[above] {\bluee{$(1-\sfe^{-K})/K$}};
  \draw (0,1) node[minimum size=0.2mm,draw,circle,  thick] {$i_1$};
 \draw (0,0) node[minimum size=0.2mm,draw,circle,   thick] {$i_2$};

       \draw (3,1) node[minimum size=0.2mm,draw,circle, thick] {$j_1$};
  \draw (3,0) node[minimum size=0.2mm,draw,circle,  thick] {$j_2$};
      \draw[dotted,   very thick] (3,-0.7)--(3,-1.3);
 \draw (3,-2) node[minimum size=0.2mm,draw,circle,   thick] {$j_K$};
\draw[-,   very thick] (0.4,0)--(2.6,0);
\draw[-,   very thick] (0.4,1)--(2.6,1);
\draw[-,   very thick] (0.4,0)--(2.6,1);
\draw[-,  very  thick] (0.4, 1)--(2.6,0);
\draw[-,  very  thick] (0.4, 1)--(2.6,-2);
\draw[-,  very  thick] (0.4, 0)--(2.6,-2);
\end{tikzpicture}
\caption{}
\label{fig:20/a/l}
  \end{subfigure}
  \hfill
\begin{subfigure}[b]{0.45\textwidth}
    \centering
\begin{tikzpicture}
  \draw (0,-1.5) node[minimum size=0.2mm,draw,circle,  thick] {$i_1$};
    \draw (3,-1.5) node[minimum size=0.2mm,draw,circle, thick] {$j_1$};
    \draw (0,-3) node[minimum size=0.2mm,draw,circle, thick] {{$i_{2}$}};
    \draw (3,-3) node[minimum size=0.2mm,draw,circle, thick] {$j_2$};
\draw[-,  very thick] (0.4,-3)--(2.6,-1.5);
\draw[-, red, ultra thick] (0.4,-1.5)--(2.6,-1.5);
\draw[-,  very thick] (0.4,-1.5)--(2.6,-3);
\draw[-,  red, ultra thick] (0.4,-3)--(2.6,-3);
\end{tikzpicture}
\caption{}
\label{fig:20/a/r}
\end{subfigure}

\caption{Two plots highlighting the positive correlation between offline nodes ($i_1$ and $i_2$) staying safe (or getting matched) at any time due to the real-time boosting in $\bt$. In Figure~\ref{fig:20/a/l}, we have a complete bipartite graph where every edge has a value of $\ep := (1 - \sfe^{-K})/K$, forming a feasible (and also optimal) solution to the natural LP in~\cite{huang2021online}. In Figure~\ref{fig:20/a/r}, we see a cycle where small edges of mass $1/3$ and big edges of mass $2/3$ are marked in black and red, respectively.}
\label{fig:positive}
\end{figure}
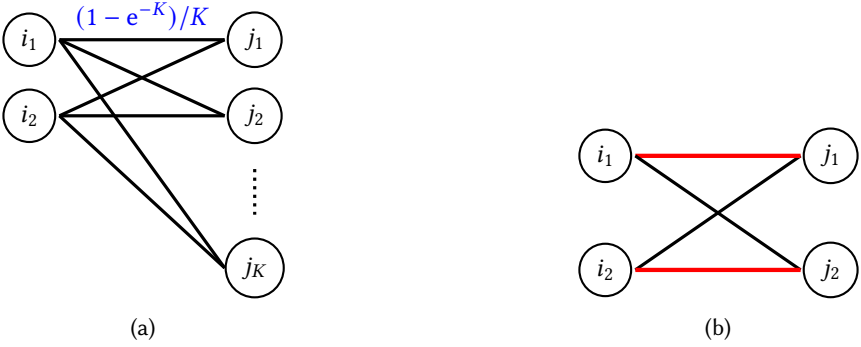

\begin{example}\label{exam:pos}
Consider the graph in Figure~\ref{fig:20/a/l}. We can verify that an optimal solution to the natural LP in~\cite{huang2021online} assigns every edge a value of $\ep = (1 - \sfe^{-K})/K$, where $K \gg 1$ is a large integer. For each $\ell = 1, 2$, let $\saf_{\ell,t} = 1$ indicate that $i_\ell$ is safe (available) at time $t \in [0,1]$ in $\bt$, and let $\saf_{\ell,t} = 0$ otherwise. Recall that every online node has a unit arrival rate. We can verify that:\footnote{This can be seen as follows: Both $i_1$ and $i_2$ are safe at time $t$ if and only if (\tbf{EV}) none of the $j_\ell$ with $1 \leq \ell \leq K$ arrives by time $t$, which occurs with probability $\sfe^{-K \cdot t}$. Similarly, $i_1$ is safe at time $t$ if and only if either (\tbf{EV}) occurs, or one of the $j_\ell$ with $1 \leq \ell \leq K$ arrives exactly once by time $t$ but happens to match with $i_2$. The same analysis applies to $i_2$ as well.}
\begin{align}
\mbox{\tbf{\emph{For Figure~\ref{fig:20/a/l}}}:} && &\E[\saf_{1,t} \cdot \saf_{2,t}] = \sfe^{-K t}, \quad \E[\saf_{1,t}] = \E[\saf_{2,t}] = \sfe^{-K t} (1 + Kt/2). \nonumber \\
&&&\frac{\E[\saf_{1,t} \cdot \saf_{2,t}]}{\E[\saf_{1,t}] \cdot \E[\saf_{2,t}]} = \frac{\sfe^{K t}}{(1 + Kt/2)^2} \ge 1. \label{ineq:pos-a}
\end{align}

For the graph in Figure~\ref{fig:20/a/r}, which is one of the WS structures for an offline node of mass one after rounding in~\cite{brubach2020online}, each offline node has two edges: one small edge of mass $1/3$ and one big edge of mass $2/3$. Similarly, we can verify that:
\begingroup
\allowdisplaybreaks
\begin{align}
\mbox{\tbf{\emph{For Figure~\ref{fig:20/a/r}}}:} && &\E[\saf_{1,t} \cdot \saf_{2,t}] = \sfe^{-2 t}, \nonumber \\
&&&\E[\saf_{1,t}] = \E[\saf_{2,t}] = \sfe^{-2 t} \left(1 + t \left(\frac{1}{3} \cdot \frac{1}{2} + \frac{2}{3} \cdot \frac{1}{2}\right)\right) = \sfe^{-2t}(1 + t). \nonumber \\
&&&\frac{\E[\saf_{1,t} \cdot \saf_{2,t}]}{\E[\saf_{1,t}] \cdot \E[\saf_{2,t}]} = \frac{\sfe^{2t}}{(1 + t)^2} \ge 1. \label{ineq:pos-b}
\end{align}
\endgroup
\end{example}

\xhdr{Remarks on Example~\ref{exam:pos}}. Expression~\eqref{ineq:pos-a} suggests the following: (i) The two offline nodes in Figure~\ref{fig:20/a/l} staying safe at any time $t \in [0,1]$ are always positively correlated; (ii) The conditional probability of one offline node staying safe given that the other is safe at any time $t$ can be inflated arbitrarily compared with the unconditional probability, \ie $\E[\saf_{1,t} ~|~ \saf_{2,t}]/\E[\saf_{1,t}] \to \infty$ as $K \to \infty$ for any given $t \in (0,1]$; and (iii) Inequality~\eqref{ineq:pos-b} implies a positive correlation between the two offline nodes staying safe \emph{at any time} in Figure~\ref{fig:20/a/r}.

\subsection{Comparison of Our Competitive Analysis Approach with That in~\cite{bib:Jaillet,brubach2020online}}\label{sec:com} 

\subsubsection{Review of the Competitive Analysis Approach in~\cite{bib:Jaillet,brubach2020online}}

To better illustrate the connections between the approach in~\cite{bib:Jaillet,brubach2020online} and ours, we introduce an auxiliary policy, denoted by $\vir$, in Algorithm~\ref{alg:vir}.

\begin{algorithm}[ht!] 
\caption{An auxiliary algorithm ($\vir$) parameterized by $\x \in \{0,1/3,2/3\}^{|E|}$.}\label{alg:vir}
\DontPrintSemicolon
\tcc{\bluee{The parameter vector $\x \in \{0,1/3,2/3\}^{|E|}$ lies in the matching polytope of the input graph $G$. We assume WLOG that $x_j := \sum_{i \in I_j} x_{ij} = 1$ for every $j \in J$ (otherwise, we can achieve this by adding dummy neighbors of $j$).}}
When an online node of type $j \in J$ arrives at time $t \in [0,1]$:\;
Let $I_{j,\x}$ be the set of nonzero neighbors of $j$ under $\x$. Generate a random list $\mathcal{L}_j$ by sampling a random permutation of $(i_1, i_2, i_3)$ over $I_{j,\x}$ such that $\Pr[\mathcal{L}_j = (i_1, i_2, i_3)] = x_{i_1,j} \cdot x_{i_2,j} / (x_{i_2,j} + x_{i_3,j})$ if $|I_{j,\x}| = 3$, and a random permutation of $(i_1, i_2)$ over $I_{j,\x}$ such that $\Pr[\mathcal{L}_j = (i_1, i_2)] = x_{i_1,j}$ if $|I_{j,\x}| = 2$. \;
Match $j$ with the first node in $\mathcal{L}_j$ \emph{that has not been matched by $j$ before}, if any; otherwise, skip $j$. \label{alg:vir/s1}\;
\tcc{\bluee{See \tbf{Remarks} on $\vir$ for clarifications on Step~\eqref{alg:vir/s1}.}}
\end{algorithm}

\xhdr{Remarks on $\vir$ in Algorithm~\ref{alg:vir}}. (i) In Step~\ref{alg:vir/s1}, an offline node $i$ on $\mathcal{L}_j$ can be matched with $j$, even if it had been previously matched before time $t$ by a neighbor $j'$ of $i$ other than $j$. Consequently, each offline node can potentially be matched with multiple distinct online neighbors. In the context of (offline) vertex-weighted scenarios, as assumed in~\cite{bib:Jaillet,brubach2020online}, our emphasis lies in establishing a lower bound for the probability of each offline node getting matched. Hence, we can confidently disregard the precise assignment of the online neighbor to it.

\smallskip

(ii) Let $\widetilde{\mathsf{AUG}}$ denote an enhanced version of $\vir$ in Algorithm~\ref{alg:vir}, where Step~\eqref{alg:vir/s1} is updated as: ``Match $j$ with the first safe node on $\mathcal{L}_j$, if any,''  where any node $i \in \mathcal{L}_j$ is considered safe if it has not been matched by any of its online neighbors. Both works~\cite{bib:Jaillet,brubach2020online} present their algorithms following a style similar to $\widetilde{\mathsf{AUG}}$. \emph{The exact version of $\vir$ outlined in Algorithm~\ref{alg:vir} is never explicitly stated in either of the two works. However, it encapsulates the core of the algorithm intensively analyzed in both papers, whose competitive ratio serves as a valid lower bound for that of the target $\widetilde{\mathsf{AUG}}$ due to Lemma~\ref{lem:vir} below}. We leave the proof of Lemma~\ref{lem:vir} to Appendix~\ref{app:vir}.

\smallskip

(iii) Although the competitive-analysis approach in both~\cite{bib:Jaillet,brubach2020online} is almost identical, the two studies adopt vastly different methods to generate the well-structured vector $\mathbf{x}$ for $\vir$: \cite{bib:Jaillet} achieves this by introducing specific constraints $0 \le x_e \le 2/3$ for every $e \in E$ to the benchmark linear program (LP), whereas~\cite{brubach2020online} applies randomized dependent rounding~\cite{gandhi2006dependent} to an optimal (fractional) solution of a benchmark LP.

\begin{lemma}[Appendix~\ref{app:vir}]\label{lem:vir}
Consider a given well-structured solution $\x \in \{0,1/3,2/3\}^{|E|}$ that lies in the matching polytope of the input graph, and let $\bt(\x)$ and $\vir(\x)$ denote the respective meta Algorithm~\ref{alg:meta} and auxiliary Algorithm~\ref{alg:vir}, each fed with $\x$. We claim that: (1) $\bt(\x)$ is equivalent to $\widetilde{\mathsf{AUG}}(\x)$, an enhanced version of $\vir$ with Step~\eqref{alg:vir/s1} updated as ``Match $j$ with the first safe node on $\mathcal{L}_j$, if any.'' (2) The performance of $\bt(\x)$ is lower bounded by that of $\vir(\x)$.
\end{lemma}

The two studies~\cite{bib:Jaillet, brubach2020online} conduct competitive analysis for $\vir$ as follows. Consider a given well-structured solution $\x \in \{0,1/3,2/3\}^{|E|}$ and an offline node $i \in I$. Both works propose a collection of \emph{mutually independent complements}, called \emph{certificates}, each characterized by a certain arrival sequence of lists that can secure $i$'s matching. 

\smallskip

Consider the instance in Figure~\ref{fig:29/a}, for example. The two offline nodes, $\bfi$ and $\hi$, share an online neighbor $j_b$, and thus, we call $\bfi$ and $\hi$ \emph{offline neighbors}.\footnote{In this case, we also say $\bfi$ offline neighbors $\hi$ through $j_b$ and vice versa.} Focus on the node $\hi$. The two studies~\cite{bib:Jaillet, brubach2020online} propose two types of certificates to compute the exact matching probability of $\hi$ in $\vir$. The first is defined as at least one arrival of any list topped by $\hi$ from $j_b$ (\tbf{CT1}); the second is defined as at least two arrivals of any list topped by $\bfi$ from $j_b$ (\tbf{CT2}). We can verify that: (1) \tbf{CT1} and \tbf{CT2} each occur with respective probabilities of $1 - \sfe^{-1/3}$ and $1 - \sfe^{-2/3}(1 + 2/3)$, and the two certificates are mutually independent complements.\footnote{To be precise, their complements exhibit asymptotic mutual independence as the time horizon $T \to \infty$.} (2) The matching of $\hi$ in $\vir(\x)$ can be guaranteed by the occurrence of either \tbf{CT1} or \tbf{CT2}, which happens with probability equal to $1 - \left(1 - (1 - \sfe^{-1/3})\right) \cdot \left(1 - (1 - \sfe^{-2/3}(1 + 2/3))\right) = 1 - \sfe^{-1}(5/3)$.

\begin{figure}
\begin{subfigure}[b]{0.45\textwidth}
    \centering
\begin{tikzpicture}
 \draw (0,0) node[minimum size=0.2mm,draw,circle, ultra thick] {$\mathbf{i}$};
   \draw (-0.4,1) node[left] {\bluee{$1/3$}};
  \draw (-0.4,0) node[left] {\bluee{$1$}};
    \draw (-0.4,-1) node[left] {\bluee{$1/3$}};
    \draw (-0.4,-2) node[left] {\bluee{$1/3$}};
 \draw (3,0) node[minimum size=0.2mm,draw,circle, thick] {$j_s$};
  \draw (3,1) node[minimum size=0.2mm,draw,circle, thick] {$j_b$};
 \draw (0,-1) node[minimum size=0.2mm,draw,circle, thick] {$\bi$};
  \draw (0,-2) node[minimum size=0.2mm,draw,circle, thick] {$\ti$};
    \draw (0,1) node[minimum size=0.2mm,draw,circle, thick] {$\hi$};

\draw[-,  very thick] (0.4,1)--(2.6,1);
\draw[-,  very thick] (0.4,0)--(2.6,0);
\draw[-,  very thick] (0.4, -1)--(2.6,0);
\draw[-,  very thick] (0.4, -2)--(2.6,0);
\draw[-,  red, ultra thick] (0.4, 0)--(2.6,1);
\end{tikzpicture}
\caption{}
\label{fig:29/a}
\end{subfigure}
\hfill
\begin{subfigure}[b]{0.45\textwidth}
    \centering
\begin{tikzpicture}
 \draw (0,0) node[minimum size=0.2mm,draw,circle, ultra thick] {$\mathbf{i}$};
   \draw (-0.4,1) node[left] {\bluee{$1/3$}};
  \draw (-0.4,0) node[left] {\bluee{$1$}};
    \draw (-0.4,-1) node[left] {\bluee{$1/3$}};
    \draw (-0.4,-2) node[left] {\bluee{$1/3$}};
 \draw (3,0) node[minimum size=0.2mm,draw,circle, thick] {$j_s$};
  \draw (3,1) node[minimum size=0.2mm,draw,circle, thick] {$j_b$};
 \draw (0,-1) node[minimum size=0.2mm,draw,circle, thick] {$\bi$};
  \draw (0,-2) node[minimum size=0.2mm,draw,circle, thick] {$\ti$};
    \draw (0,1) node[minimum size=0.2mm,draw,circle, thick] {$\hi$};

\draw[-,  very thick] (0.4,1)--(2.6,1);
\draw[-,  very thick] (0.4,0)--(2.6,0) node [blue, above, near end, sloped] {$1 - 2z_2$};
\draw[-,  very thick] (0.4, -1)--(2.6,0) node [blue, above, midway, sloped] {$z_2$};
\draw[-,  very thick] (0.4, -2)--(2.6,0) node [blue, below, midway, sloped] {$z_2$};
  \draw (1.5,1) node[above] {\bluee{$z_1$}};
\draw[-,  red, ultra thick] (0.4, 0)--(2.6,1) node [blue, above, near start, sloped] {$1 - z_1$};
\end{tikzpicture}
\caption{}
\label{fig:29/b}
\end{subfigure}

\caption{An example highlighting (i) the need for modifications to sampling distributions of online nodes and (ii) the difference between our approach and those in~\cite{bib:Jaillet, brubach2020online}. \emph{Throughout this paper, we assume the following unless specified otherwise: (i) Big edges of mass $2/3$ and small edges of mass $1/3$ are marked in red and black, respectively; (ii) The value next to each offline node represents its total mass in the well-structured vector $\x \in \{0,1/3,2/3\}^{|E|}$.} \\
Figure~\ref{fig:29/a} shows that even under the current holistic competitive analysis, the Matching Probability per Mass (MPM) achieved by node $\mathbf{i}$ in $\bt$ equals $1 - 22/(9 \sfe^2) \approx 0.6692 < \kap^*_B$, the target MPM for an offline node of mass one in the form of 1B1S (see Appendix~\ref{app:mot/one}). \\ Figure~\ref{fig:29/b} proposes modified sampling distributions for nodes $j_b$ and $j_s$, where the input vector for $j_b$ is updated from $(1/3, 2/3)$ to $(z_1, 1 - z_1)$ with $z_1 \in [0,1]$. Similarly, the input vector for $j_s$ is updated from $(1/3, 1/3, 1/3)$ to $(1 - 2z_2, z_2, z_2)$, where $z_2 \in [0,1/2]$. In Section~\ref{sec:a+c}, we establish that an \emph{aggressive} setting of $z_1 = z_2 = 0$ in Figure~\ref{fig:29/b} suffices to ensure every offline node achieves an MPM greater than the target specified in Proposition~\ref{pro:main-1}. Specifically, we show that under the \emph{aggressive} setting, node $\bfi$ achieves an MPM equal to $1 - \sfe^{-2} > \kapc$, and nodes $\bi$ and $\ti$ each achieve an MPM of $3(1 - 2/\sfe) \approx 0.7927 > \kapa$. This contrasts with the fact that in the same aggressive setting, nodes $\bi$ and $\ti$ each achieve an MPM equal to $3\left(1 - 9/(4\sfe)\right) \approx 0.5168 < \kapa$ following  the approach in~\cite{bib:Jaillet, brubach2020online}.\protect\footnotemark}
\label{fig:12/27/a}
\end{figure}
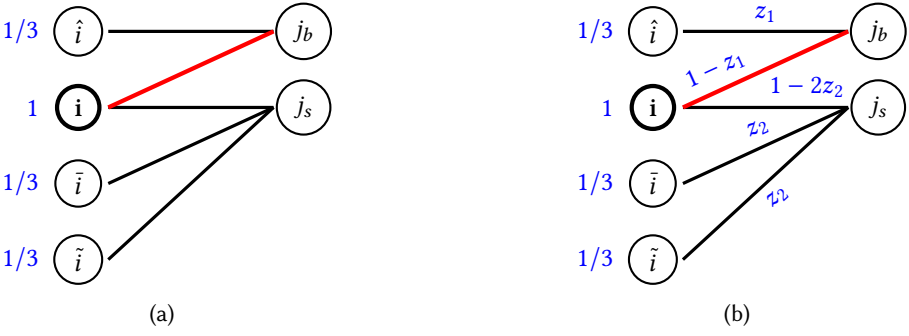
\footnotetext{This can be seen as follows: By the approach in~\cite{bib:Jaillet, brubach2020online}, when $z_2 = 0$, a list topped by $\bfi$, denoted by $\mathcal{L} = (\bfi, *, *)$, gets sampled with probability one when $j_s$ arrives. Thus, the total number of arrivals of $\mathcal{L}$ is $\text{Pois}(1)$ (a Poisson random variable with mean one). Consequently, for each of $\bi$ and $\ti$, it gets matched with probability 0, $1/2$, and 1 when $\mathcal{L}$ arrives once, twice, and at least three times, respectively. For more details, see Lemma~\ref{lem:7/3/b} and its proof in Appendix~\ref{app:mot/one}.}

\smallskip
\emph{We emphasize that the approach proposed in~\cite{bib:Jaillet, brubach2020online} focuses on a specific class of certificates, each characterized by a sequence of arriving lists that involve only a single online neighbor.} This method disregards certificates involving multiple online neighbors that could still guarantee the matching of the target offline node. \tbf{For the case of $\hi$ in Figure~\ref{fig:29/a}}: In addition to the two certificates previously mentioned, another candidate that can secure the matching of $\hi$ is as follows: a sequence consisting of a list that is {topped} by $\bfi$ but associated with the online neighbor $j_s$ of $\bfi$, followed by another list topped by $\bfi$ from $j_b$ (\textbf{CT3}). Note that Certificate \textbf{CT3} ensures the matching of $\hi$ in $\widetilde{\mathsf{AUG}}$ but not in $\vir$.\footnote{This is why we assert that both works of~\cite{bib:Jaillet, brubach2020online} essentially analyze a weaker version ($\vir$) of the algorithm proposed in the paper ($\widetilde{\mathsf{AUG}}$).}

\smallskip
The approach that focuses on certificates involving a single online neighbor~\cite{bib:Jaillet, brubach2020online} has both advantages and disadvantages. An advantage is that it enables us to circumvent the challenge posed by the positive correlations among offline nodes getting matched in $\widetilde{\mathsf{AUG}}$. Notably, since certificates involving a single online neighbor prove mutually independent complements, the task of lower bounding the matching probability of any target offline node is greatly simplified: We only need to identify all possible \emph{certificates} and evaluate their respective probabilities. This immediately leads to a valid lower bound, as we did for the instance in Figure~\ref{fig:29/a}. However, a notable disadvantage is apparent: The obtained result offers only a lower bound on the exact matching probability, which potentially bears a considerable gap due to miscounting certificates that \emph{involve} multiple online neighbors but can still ensure the matching of any target offline node.

\subsubsection{An Alternative Interpretation of the Approach in~\cite{bib:Jaillet, brubach2020online}}\label{sec:int}

We offer insights into $\vir$ (Algorithm~\ref{alg:vir}), the essential algorithm analyzed in both~\cite{bib:Jaillet, brubach2020online}, from the perspective of the real-time boosting impact on the sampling distributions of online nodes. As clarified in ``\textbf{Remarks on $\vir$}'' on page 7,  upon the arrival of an online node $j$ at time $t$, $\vir$ disregards the real-time status of being safe or matched for each $i \in I_{j,\x}$. Instead, it solely monitors whether each $i \in I_{j,\x}$ was matched by $j$ itself before. Consequently, this might result in wasting $j$ by matching it with some $i \in I_{j,\x}$ that had been matched before by some $j' \neq j$, where $j' \in J_{i,\x}$ (the set of nonzero neighbors of $i$ under $\x$).

\smallskip

In terms of boosting effects, $\vir$ only \emph{partially} exploits the benefit introduced by real-time boosting to the sampling distribution $\mathcal{D}_{j,t}$ for $j$ at time $t$: $\mathcal{D}_{j,t}$ will get boosted only when some neighbor $i \in I_{j,\x}$ gets matched by $j$ itself. In other words, $\mathcal{D}_{j,t}$ would remain invariant even when some $i \in I_{j,\x}$ is matched by some $j' \neq j$ with $j' \in J_{i,\x}$. This stands in contrast to $\bt$ and $\widetilde{\mathsf{AUG}}$, which \emph{fully} harness the power of real-time boosting. In these cases, $\mathcal{D}_{j,t}$ receives a boost whenever any $i \in I_{j,\x}$ is matched before time $t$, regardless of the matching agent for $i$.

\subsubsection{Our Approach}

In this paper, we present a \emph{holistic} approach to evaluate the matching probability of a target offline node. Specifically, we propose and harness the power of Ordinary Differential Equations (ODEs) systems to capture the real-time boosting impact on the sampling distributions of online nodes. We use the instance in Figure~\ref{fig:29/b} to illustrate the differences.

As shown in~\cite{bib:Jaillet, brubach2020online}, the bottleneck arises at offline nodes of mass one rather than those with mass less than one. As a result, we need to reduce the performance on offline nodes of mass less than one to compensate for those of mass one by adding modifications to the sampling distributions of online nodes. For the case in Figure~\ref{fig:29/b}, we adjust the sampling distribution on $j_s$ by updating the input vector from $(1/3, 1/3, 1/3)$ to $(1 - 2z_2, z_2, z_2)$, where $z_2 \in [0, 1/2]$. The default setting is $z_2 = 1/3$, and we aim to identify the \emph{smallest} possible value of $z_2$ such that $\bi$ and $\ti$ each achieve a Matching Probability per Mass (MPM) equal to a preset target $\kap$, thereby benefiting node $\bfi$ the most.

\smallskip

In Appendix~\ref{app:mot/one}, we show that under the approach in~\cite{bib:Jaillet, brubach2020online}, offline nodes $\bi$ and $\ti$ in Figure~\ref{fig:29/b} each achieve an MPM of
\begin{align*}
\eta(z_2) := \frac{3 \left(8 z_2^2 + \frac{-14 z_2^2 + 19 z_2 + 9}{\sfe} - 12 z_2 - 4\right)}{4 (z_2 - 1)}.
\end{align*}
The work by~\cite{brubach2020online} asserted an MPM of $\tkap = 0.7622$ for any offline node of mass $1/3$. Consequently, to ensure that $\bi$ and $\ti$ in Figure~\ref{fig:29/b} each achieve an MPM of at least $\tkap$, we must set $\eta(z_2) \geq \kappa$, which can be solved as $z_2 \geq 0.0558$. This aligns with the configuration of $z_2 = 0.1$, as proposed in~\cite{brubach2020online}. However, it contrasts with our choice of $z_2 = 0$ for the same instance in Figure~\ref{fig:29/b}.

\smallskip

In Section~\ref{sec:a+c}, we establish that $z_1 = z_2 = 0$ suffices to guarantee that every offline node of mass $1/3$ achieves an MPM as high as $0.7927$. The improvement arises from a distinct approach based on the Ordinary Differential Equations (ODEs) system: it allows for a \emph{comprehensive} competitive analysis, accounting not only for certificates involving a single online node, as done in~\cite{bib:Jaillet, brubach2020online}, but also for those involving multiple different online nodes (\ie $j_b$ and $j_s$) when assessing the matching probability of $\bi$ and $\ti$. Put in the context of boosting, the holistic ODEs system-based approach enables us to directly analyze $\bt$ (or $\widetilde{\mathsf{AUG}}$), which enjoys the full benefits brought by real-time boosting to the sampling distributions of online nodes.

\subsection{Other Related Work}

A significant body of work explores various variants of online matching models, as detailed in the survey book by~\cite{mehta2012online}. In this discussion, we narrow our focus to studies specifically investigating the model of Online Matching under Known Identical and Independent Distributions (OM-KIID). The seminal work by~\cite{bib:Feldman} introduced the first algorithm for (unweighted) OM-KIID, achieving a competitiveness that notably exceeds the golden barrier of $1 - 1/\sfe$. Subsequently, several works have examined OM-KIID under different objectives, including vertex-weighted and edge-weighted settings, as explored by~\cite{bib:Jaillet,bib:Manshadi,bib:Haeupler,tang2022fractional}. Notably,~\cite{yan2024edge} presented an algorithm achieving a competitive ratio of 0.645, the first to strictly surpass $1 - 1/\sfe$ for \emph{edge-weighted} OM-KIID \emph{with general arrival rates}, which was recently improved to 0.650 by~\cite{qiu2023improved}.

\xhdr{Other Hardness Results Concerning OM-KIID}. For vertex-weighted OM-KIID \emph{with general arrival rates}, the best upper bound is 0.823, as established by~\cite{bib:Manshadi}. For \emph{edge-weighted} OM-KIID \emph{with general arrival rates},~\cite{huang2022power} provided an improved bound of 0.703, based on a delicate (online-sided) vertex-weighted instance. Additionally,~\cite{ma2021grouplevel} considered OM-KIID \emph{with general arrival rates} under the objective of fairness maximization among online agents, offering a hardness result of $\sqrt{3} - 1 \approx 0.7321$.

\xhdr{Studies Related to Real-Time Boosting}. The idea of real-time boosting is natural and has been implemented and empirically evaluated as a heuristic in various real-world matching markets~\cite{aamas-19,dickerson2018assigning}.~\cite{ma2023fairness} conducted a formal analysis of algorithms incorporating real-time boosting for the OM-KIID model with integral arrival rates \emph{but focused on maximizing fairness among offline agents}. They introduced a matching policy employing boosting that achieves a competitiveness of 0.722. Notably, the objective of fairness maximization among offline agents implies an exclusive property for any optimal solution $\x$ to a benchmark LP: All offline agents can be assumed, without loss of generality, to have uniform mass under $\x$. This assumption does not apply to our case, however, as we aim to maximize the expected total weight among all matched offline nodes. The study by~\cite{ma2023fairness} further demonstrates that, with minor adjustments, the same algorithm achieves a competitiveness of 0.719 for vertex-weighted OM-KIID with integral rates, the exact model considered in this paper. Unfortunately, the ratio of 0.719 presented there does not improve either of the two state-of-the-art results of 0.7299 or 0.725, established by~\cite{brubach2020online} and~\cite{bib:Jaillet}, respectively.

\section{When All Offline Nodes Have a Mass of One}\label{sec:one}

In this section, we prove \emph{Claim (1)} in Proposition~\ref{pro:main-1} \emph{in the context where all offline nodes have a mass of one (on the random rounded vector)}. The claim states that for any realization of $\X$, $\bt(\X)$ achieves a \emph{Matching Probability per Mass} (MPM) of at least $\kapc = 1 - 2\sfe^{-2}$ and $\kapd = 1 - \frac{9}{2} \cdot \sfe^{-3}$ for any offline node of mass one, when instantiated as 1B1S (one big and one small edge) and as 3S (three small edges), respectively.

\subsection{\tbf{Folding Procedure} and \re}\label{sec:fold}

Consider the instance in Figure~\ref{fig:11/15/23/a}, which illustrates the \tbf{Folding Procedure} applied to an offline node of mass one when all its offline neighbors also have a mass of one.

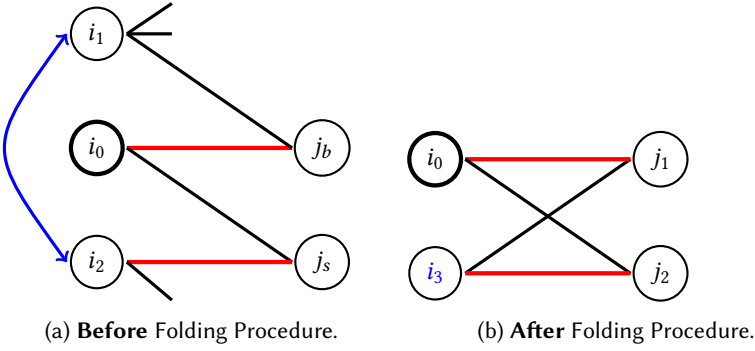
\begin{figure}[ht!]
\begin{subfigure}[b]{0.4\textwidth}
\begin{tikzpicture}
\draw[-, very thick] (0.4,0)--(1,0.4);
\draw[-, very thick] (0.4,0)--(1,0);
\draw (0,0) node[minimum size=0.2mm,draw,circle, thick] {$i_1$};
\draw (0,-1.5) node[minimum size=0.2mm,draw,circle, ultra thick] {$i_0$};
\draw (3,-1.5) node[minimum size=0.2mm,draw,circle, thick] {$j_b$};
\draw (0,-3) node[minimum size=0.2mm,draw,circle, thick] {$i_2$};
\draw (3,-3) node[minimum size=0.2mm,draw,circle, thick] {$j_s$};
\draw[-, red, ultra thick] (0.4,-1.5)--(2.6,-1.5);
\draw[-, very thick] (0.4,-1.5)--(2.6,-3);
\draw[-, red, ultra thick] (0.4,-3)--(2.6,-3);
\draw[-, very thick] (0.4,-3)--(1,-3.5);
\draw[-, very thick] (0.4,0)--(2.6,-1.5);
\draw[<->, blue, very thick] (-0.4,0) .. controls (-1.5,-1.5) .. (-0.4,-3);
\end{tikzpicture}
\caption{\tbf{Before} Folding Procedure.}
\label{fig:11/15/23/a/l}
\end{subfigure}
\begin{subfigure}[b]{0.4\textwidth}
\begin{tikzpicture}
\draw (0,-1.5) node[minimum size=0.2mm,draw,circle, ultra thick] {$i_0$};
\draw (3,-1.5) node[minimum size=0.2mm,draw,circle, thick] {$j_1$};
\draw (0,-3) node[minimum size=0.2mm,draw,circle, thick] {\bluee{$i_{3}$}};
\draw (3,-3) node[minimum size=0.2mm,draw,circle, thick] {$j_2$};
\draw[-, very thick] (0.4,-3)--(2.6,-1.5);
\draw[-, red, ultra thick] (0.4,-1.5)--(2.6,-1.5);
\draw[-, very thick] (0.4,-1.5)--(2.6,-3);
\draw[-, red, ultra thick] (0.4,-3)--(2.6,-3);
\end{tikzpicture}
\caption{\tbf{After} Folding Procedure.}
\label{fig:11/15/23/a/r}
\end{subfigure}

\caption{The target offline node $i_0$ is in the form of 1B1S (one big and one small edge, marked in red and black, respectively), where the small online neighbor $j_s$ has another big edge $(i_2, j_s)$. The \tbf{Folding Procedure} merges the two offline neighbors of $i_0$, \ie $i_1$ and $i_2$, into one (\bluee{$i_3$}), resulting in the structure shown on the right.}
\label{fig:11/15/23/a}
\end{figure}

\begin{lemma}\label{lem:11/15/a}
After the \tbf{Folding Procedure}, the probability that the target offline node $i_0$ remains safe in Figure~\ref{fig:11/15/23/a/r} is no less than in Figure~\ref{fig:11/15/23/a/l} at any time.
\end{lemma}

\begin{proof}

Consider the discrete version of the arrival setting. Focus on the structure in Figure~\ref{fig:11/15/23/a/l}, and let $\alpha_k$ be the probability that $i_0$ stays safe at round $k \in [T]$. Let $\beta_{\ell,k} = \E[\saf_{i_0,k} \cdot \saf_{i_\ell,k}]$ for $\ell = 1,2$, which represents the probability that both $i_0$ and $i_\ell$ are safe at round $k \in [T]$. Observe that $\beta_{\ell,k} / \alpha_k = \E[\saf_{i_\ell,k} | \saf_{i_0,k}]$, the conditional probability that $i_\ell$ is safe at round $k$ given $i_0$ is safe at that time.

Assume $i_0$ is safe at some round $k \in [T]$. The matching rate of $i_0$ from $j_b$ is $2/3$ given $i_1$ is safe, and it is $1$ if $i_1$ is not safe. Thus, assuming $i_0$ is safe at round $k$, the total matching rate of $i_0$ from $j_b$ in Figure~\ref{fig:11/15/23/a/l} is equal to
\[
\left(\frac{2}{3}\right) \left(\frac{\beta_{1,k}}{\alpha_k}\right) + 1 - \frac{\beta_{1,k}}{\alpha_k} = 1 - \frac{1}{3} \left(\frac{\beta_{1,k}}{\alpha_k}\right).
\]
We can derive the matching rate of $i_0$ from $j_s$ similarly. Therefore, the dynamics of the series $(\alpha_k)_k$ is captured as follows:
\begin{align*}
\alpha_{k+1} &= \alpha_k \cdot \left(1 - \frac{1}{T} \left[\frac{2}{3} \cdot \frac{\beta_{1,k}}{\alpha_k} + 1 - \frac{\beta_{1,k}}{\alpha_k} + \frac{1}{3} \cdot \frac{\beta_{2,k}}{\alpha_k} + 1 - \frac{\beta_{2,k}}{\alpha_k}\right]\right) \\
&= \alpha_k \cdot \left(1 - \frac{1}{T} \left[2 - \frac{1}{3} \cdot \frac{\beta_{1,k}}{\alpha_k} - \frac{2}{3} \cdot \frac{\beta_{2,k}}{\alpha_k}\right]\right).
\end{align*}

Similarly, let $\tilde{\alpha}_k$ be the probability that $i_0$ stays safe at round $k \in [T]$, and let $\beta_{3,k}$ represent the probability that both $i_0$ and $i_3$ are safe at round $k \in [T]$ in Figure~\ref{fig:11/15/23/a/r}. We have:

\begin{align*}
\tilde{\alpha}_{k+1} &= \tilde{\alpha}_k \cdot \left(1 - \frac{1}{T} \left[\frac{2}{3} \cdot \frac{\beta_{3,k}}{\alpha_k} + 1 - \frac{\beta_{3,k}}{\alpha_k} + \frac{1}{3} \cdot \frac{\beta_{3,k}}{\alpha_k} + 1 - \frac{\beta_{3,k}}{\alpha_k}\right]\right) \\
&= \tilde{\alpha}_k \cdot \left(1 - \frac{1}{T} \left[2 - \frac{\beta_{3,k}}{\alpha_k}\right]\right).
\end{align*}

Now, let us shift to the continuous arrival setting. Let $\alpha(t)$, $\tilde{\alpha}(t)$, $\beta_1(t)$, $\beta_2(t)$, and $\beta_3(t)$ be the continuous counterparts of $\alpha_k$, $\tilde{\alpha}_k$, $\beta_{1,k}$, $\beta_{2,k}$, and $\beta_{3,k}$, respectively, with $t \in [0,1]$. By taking $T \to \infty$, we have:
\begin{align*}
\alpha'(t) &= -2 \alpha(t) + \frac{1}{3} \beta_1(t) + \frac{2}{3} \beta_2(t), \quad \alpha(0) = 1; \\
\tilde{\alpha}'(t) &= -2 \tilde{\alpha}(t) + \beta_3(t), \quad \tilde{\alpha}(0) = 1.
\end{align*}

Note that at any time $t \in [0,1]$, we have: (1) $\beta_1(t) \leq \sfe^{-2t}$ and $\beta_2(t) \leq \sfe^{-2t}$. For the pair of nodes $i_0$ and $i_1$, the sum of their matching rates should be at least 2 at any time $t \in [0,1]$ if both are safe, and it could be strictly larger than 2 when either $i_2$ is matched (thus boosting the matching rate from $j_s$ from $1/3$ to 1) or any offline neighbor of $i_1$ other than $i_0$ is matched. This leads to $\beta_1(t) \leq \sfe^{-2t}$. We can argue similarly for $\beta_2(t) \leq \sfe^{-2t}$.
(2) $\beta_3(t) = \sfe^{-2t}$, since both $i_0$ and $i_3$ are safe iff neither $j_1$ nor $j_2$ ever arrives during $[0,t]$.

Thus, we conclude that $\tilde{\alpha}'(t) \geq \alpha'(t)$ for all $t \in [0,1]$. Since both functions have the same initial value at $t = 0$, we claim $\tilde{\alpha}(t) \geq \alpha(t)$ for all $t \in [0,1]$.
\end{proof}

The result in Lemma~\ref{lem:11/15/a} on the two structures shown in Figures~\ref{fig:11/15/23/a/l} and~\ref{fig:11/15/23/a/r} suggests that the WS structures must share some kind of symmetry. We formally state it below.

\begin{tcolorbox}
\xhdr{\re}. Consider a target offline node $\bfi$, and suppose it is in a Worst-Scenario (WS) structure such that the probability that $\bfi$ stays safe is maximized in \bt. We claim that $\bfi$'s WS structure should be instantiated when all of its offline neighbors are also in the worst possible structure, such that each has the largest possible probability of staying safe as well.
\end{tcolorbox}

In general, \re suggests that WS structures do not exist in isolation but arise simultaneously and exhibit shared symmetrical characteristics. This observation stems from the fact that any decreased performance among the offline neighbors of the target node $\bfi$, such as a reduced matching probability or an increased safe probability at any given time, would invariably impede the matching process for $i$. Consequently, this exacerbates the performance of its neighboring offline nodes. We offer a formal proof of the \re in Appendix~\ref{app:re}.

\subsection{Three Possible Worst-Scenario (WS) Structures for an Offline Node of Mass One}\label{sec:wsone}

By repeatedly applying the \tbf{Folding Procedure} and the \re, we identify the following three possible WS structures for an offline node of mass one, as stated in the lemma below. By default, we assume that all offline and online nodes have a mass of one after rounding.

\begin{lemma}\label{lem:ws_one}
When all offline nodes have a unit mass (after dependent rounding), repeatedly applying the \tbf{Folding Procedure} and the \tbf{Symmetry Principle} yields only three possible WS structures for a mass-one offline node, as illustrated in Figure~\ref{fig:three}.
\end{lemma}

In Appendix~\ref{app:fold}, we provide a formal definition of the \tbf{Folding Procedure} and a proof of Lemma~\ref{lem:ws_one}.

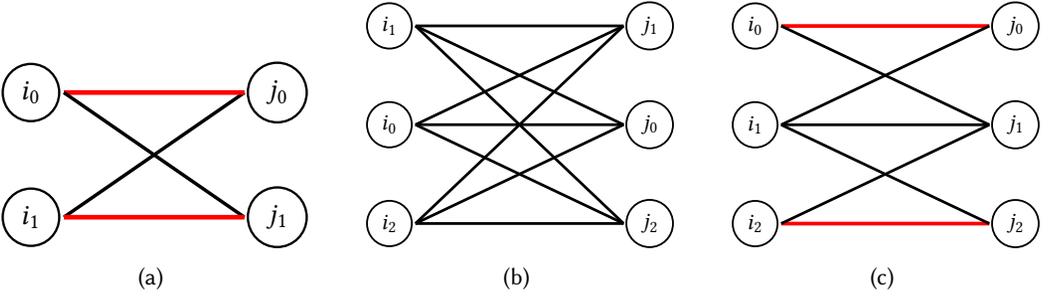
\begin{figure}[ht!]
\begin{subfigure}[b]{0.3\textwidth}
        \centering
        \resizebox{\linewidth}{!}{
\begin{tikzpicture}
  \draw (0,-1.5) node[minimum size=0.2mm,draw,circle,  thick] {$i_0$};
    \draw (3,-1.5) node[minimum size=0.2mm,draw,circle, thick] {$j_0$};
    \draw (0,-3) node[minimum size=0.2mm,draw,circle, thick] {{$i_{1}$}};
    \draw (3,-3) node[minimum size=0.2mm,draw,circle, thick] {$j_1$};
\draw[-,  very thick] (0.4,-3)--(2.6,-1.5);
\draw[-, red, ultra thick] (0.4,-1.5)--(2.6,-1.5);
\draw[-,  very thick] (0.4,-1.5)--(2.6,-3);
\draw[-,  red, ultra thick] (0.4,-3)--(2.6,-3);
\end{tikzpicture}}
\caption{}  \label{fig:three/a}
\end{subfigure}
\hfill 
\begin{subfigure}[b]{0.3\textwidth}
        \centering
        \resizebox{\linewidth}{!}{
 \begin{tikzpicture}
  \draw (0,1.5) node[minimum size=0.2mm,draw,circle, thick] {$i_1$};
 \draw (4,1.5) node[minimum size=0.2mm,draw,circle, thick] {$j_1$};
  \draw (0,0) node[minimum size=0.2mm,draw,circle, thick] {$i_0$};
  \draw (0,-1.5) node[minimum size=0.2mm,draw,circle, thick] {$i_2$};
  \draw (4,0) node[minimum size=0.2mm,draw,circle, thick] {$j_0$};
   \draw (4,-1.5) node[minimum size=0.2mm,draw,circle, thick] {{$j_2$}};

             \draw[-, very thick ] (0.4,0)--(3.6,0);   
\draw[-,  very thick] (0.4,0)--(3.6,-1.5);
\draw[-,  very thick] (0.4,1.5)--(3.6,0);
\draw[-,  very thick] (0.4,-1.5)--(3.6,-1.5);
\draw[-,  very thick] (0.4,1.5)--(3.6,1.5);
\draw[-,  very thick] (0.4,1.5)--(3.6,-1.5);
\draw[-,  very thick] (0.4,0)--(3.6,1.5);
\draw[-,  very thick] (0.4,-1.5)--(3.6,1.5);
\draw[-,  very thick] (0.4,-1.5)--(3.6,0);
\end{tikzpicture}}
\caption{}  \label{fig:three/b}
    \end{subfigure}  
\hfill
\begin{subfigure}[b]{0.3\textwidth}
        \centering
        \resizebox{\linewidth}{!}{
 \begin{tikzpicture}
 \draw (0,4.5) node[minimum size=0.2mm,draw,circle, thick] {$i_0$};
  \draw (0,3) node[minimum size=0.2mm,draw,circle, thick] {$i_1$};
    \draw (0,1.5) node[minimum size=0.2mm,draw,circle, thick] {$i_2$};
  \draw (4,4.5) node[minimum size=0.2mm,draw,circle, thick] {$j_0$};
   \draw (4,3) node[minimum size=0.2mm,draw,circle, thick] {{$j_1$}};
      \draw (4,1.5) node[minimum size=0.2mm,draw,circle, thick] {{$j_2$}};
             \draw[-, ultra thick, red ] (0.4,4.5)--(3.6,4.5);   
               \draw[-, ultra thick, red ] (0.4,1.5)--(3.6,1.5);   

\draw[-,  very thick] (0.4,4.5)--(3.6,3);
\draw[-,  very thick] (0.4,3)--(3.6,4.5);
\draw[-,  very thick] (0.4,3)--(3.6,3);
\draw[-,  very thick] (0.4,3)--(3.6,1.5);
\draw[-,  very thick] (0.4,1.5)--(3.6,3);
\end{tikzpicture}}
\caption{}  \label{fig:three/c}
\end{subfigure}
  \caption{Three possible WS structures for an offline node of mass one after the rounding, where big edges of mass $2/3$ and small edges of mass $1/3$ are marked in red and black, respectively.}
    \label{fig:three}
\end{figure}

\begin{table}[ht!]
\caption{Summary of the exact matching probabilities of offline nodes in the WS structures shown in Figures~\ref{fig:three/a}, ~\ref{fig:three/b}, and ~\ref{fig:three/c}, where all nodes have a unit mass after rounding. The type (1B1S) represents an offline node with one big and one small edge, while (3S) represents a node with three small edges. All fractional values are rounded to the fourth decimal place, if necessary, and the smallest value in each row is marked in blue.}\label{table:sum/three}
\begin{tabular}{c " ccc} \thickhline 
& Figure~\ref{fig:three/a} & Figure~\ref{fig:three/b} & Figure~\ref{fig:three/c} \\ \thickhline 
(1B1S) & \bluee{0.7293} & & 0.7314 \\
(3S) & & \bluee{0.7760} & 0.7925 \\ \thickhline 
\end{tabular}
\end{table}

\subsubsection{Computation of Matching Probabilities in Figures~\ref{fig:three/a} and~\ref{fig:three/b}}
The computation of the exact matching probabilities in Figures~\ref{fig:three/a} and~\ref{fig:three/b} is simplified by their highly symmetric structures. Consider the case in Figure~\ref{fig:three/a} first, and let $N$ be the total number of matches over $i_0$ and $i_1$. We observe that:
\begin{align*}
\E[N] &= \Pr[N \ge 1] + \Pr[N \ge 2] = 1 - \sfe^{-2} + 1 - \sfe^{-2}(1 + 2) = 2 - 4 \sfe^{-2}.
\end{align*}
Thus, we claim that $i_0$ and $i_1$ each have a matching probability equal to $(2 - 4 \sfe^{-2}) / 2 = 1 - 2 \sfe^{-2} \approx 0.7293$. Similarly, we find that each offline node in Figure~\ref{fig:three/b} has a matching probability equal to $1 - \frac{9}{2} \sfe^{-3} \approx 0.7760$.

\subsubsection{Computation of Matching Probabilities in Figure~\ref{fig:three/c}} \label{sec:comp}

Below is an Ordinary Differential Equations (ODEs)-based approach. In Appendix~\ref{app:three/c}, we present another Markov-chain-based method, which yields the same results.

Recall that the continuous version of our arrival setting states that each online node arrives following an independent Poisson process with rate one over the time range $[0,1]$. For any time $t \in [0,1]$, let $\alp(t)$, $\beta(t)$, and $\gam(t)$ denote the probability \anhai{that} at time $t$, $i_1$ is safe, $i_1$ is safe and exactly one of $i_0$ and $i_2$ is safe, and all three are safe, respectively. To derive the dynamics among $\alp$, $\beta$, and $\gam$, we first examine their discrete counterparts. Set $\alp_k = \alp(k/T)$ for each round $k \in [T]$, and similarly define $\beta_k$ and $\gam_k$. Observe that:
\begingroup
\allowdisplaybreaks
\begin{align*}
\alp_{k+1} &= \alp_k \cdot \left(1 - \frac{1}{T} \left( \frac{\beta_k}{\alp_k} \cdot \left(\frac{1}{3} + \frac{1}{2} + 1\right) + \frac{\gam_k}{\alp_k} + \left(1 - \frac{\beta_k}{\alp_k} - \frac{\gam_k}{\alp_k}\right) \cdot 3 \right) \right), && \alp_0 = 1; \\
\beta_{k+1} &= \beta_k \cdot \left(1 - \frac{3}{T}\right) + \gam_k \cdot \frac{2}{T}, && \beta_0 = 0; \\
\gam_{k+1} &= \gam_k \cdot \left(1 - \frac{3}{T}\right), && \gam_0 = 1,
\end{align*}
\endgroup

where $\beta_k / \alp_k$ represents the conditional probability that exactly one of $i_1$'s offline neighbors is safe at round $k$, given that $i_1$ is safe, and $\gam_k / \alp_k$ represents the conditional probability that both of $i_1$'s offline neighbors are safe at round $k$, given that $i_1$ is safe. Converting the above discrete version to the continuous case by taking $T$ to infinity, we obtain:
\begin{align*}
\alp' &= -3 \alp + \frac{7}{6} \beta + 2 \gam, && \alp(0) = 1; \\
\beta' &= -3 \beta + 2 \gam, && \beta(0) = 0; \\
\gam' &= -3 \gam, && \gam(0) = 1.
\end{align*}
We can solve the above ODE system and obtain:
\[
\gam(t) = \sfe^{-3t}, \quad \beta(t) = 2t \cdot \sfe^{-3t}, \quad \alp(t) = \sfe^{-3t} \cdot \left(1 + 2t + \frac{7}{6} t^2\right), \quad \forall t \in [0,1].
\]

Thus, $i_1$ gets matched with a probability equal to $1 - \alp(1) = 1 - \sfe^{-3} \cdot \frac{25}{6}$, which is consistent with what we obtained previously. A similar approach can be applied to compute the matching probabilities for $i_0$ and $i_2$ as well.

\section{Worst-Scenario (WS) Principle for the General Case}\label{sec:less}

For the general case when offline nodes do not necessarily have a unit mass after rounding, it becomes easier to pinpoint the WS structures compared with the case when all offline nodes have a mass of one. Recall that for an offline node $i$, another offline node $i'$ is considered an offline neighbor of $i$ if they share at least one common online neighbor $j$ with $(i,j) \in E$ and $(i',j) \in E$.

\begin{tcolorbox}
\xhdr{Worst-Scenario (WS) Principle}. Consider an offline node $\bfi$. We claim that $\bfi$'s WS structure (when it has the largest probability of staying safe in \bt) must be instantiated when every offline neighbor of $\bfi$ has a single online neighbor, which is the exact one shared with $\bfi$.
\end{tcolorbox}

The principle above can be viewed as an application of the \re to the general case. For an offline node $\bfi$, the validity of the principle can be seen as follows:
\begin{enumerate}
    \item Any of its offline neighbors must share a single online neighbor with it; otherwise, we can decompose it into a strictly worse structure, as illustrated in Figure~\ref{fig:dec}.
    \item Any of its offline neighbors $\ti$ must have a single online neighbor, say $j$, that is the exact one shared with $\bfi$ itself; otherwise, we can prune any extra edges connecting $\ti$ with some online neighbors $\tj \neq j$, which would strictly worsen the performance of $\ti$ (i.e., increase the probability of $\ti$ staying safe), and consequently, worsen that of $\bfi$ as well.
\end{enumerate}

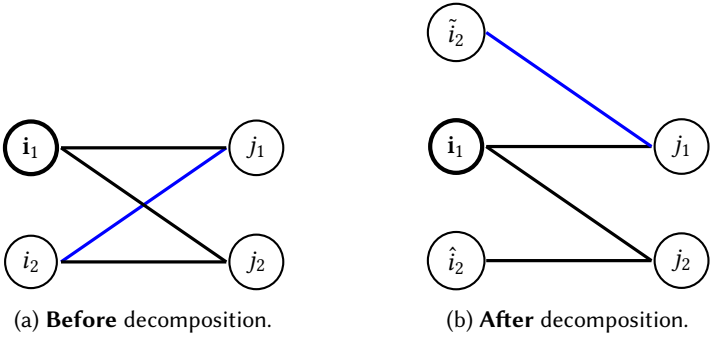
\begin{figure}[ht!]
\begin{subfigure}[b]{0.4\textwidth}
\centering
\begin{tikzpicture}
\draw (0,-1.5) node[minimum size=0.2mm,draw,circle, ultra thick] {$\bfi_1$};
\draw (3,-1.5) node[minimum size=0.2mm,draw,circle, thick] {$j_1$};
\draw (0,-3) node[minimum size=0.2mm,draw,circle, thick] {{$i_{2}$}};
\draw (3,-3) node[minimum size=0.2mm,draw,circle, thick] {$j_2$};
\draw[-,  blue, very thick] (0.4,-3)--(2.6,-1.5);
\draw[-,  very thick] (0.4,-1.5)--(2.6,-1.5);
\draw[-,  very thick] (0.4,-1.5)--(2.6,-3);
\draw[-,  very thick] (0.4,-3)--(2.6,-3);
\end{tikzpicture}
\caption{\tbf{Before} decomposition.}
\label{fig:dec/a}
\end{subfigure}
\begin{subfigure}[b]{0.4\textwidth}
\centering
\begin{tikzpicture}
\draw (0,0) node[minimum size=0.2mm,draw,circle, thick] {$\ti_2$};
\draw (0,-1.5) node[minimum size=0.2mm,draw,circle, ultra thick] {$\bfi_1$};
\draw (3,-1.5) node[minimum size=0.2mm,draw,circle, thick] {$j_1$};
\draw (0,-3) node[minimum size=0.2mm,draw,circle, thick] {{$\hi_{2}$}};
\draw (3,-3) node[minimum size=0.2mm,draw,circle, thick] {$j_2$};
\draw[-, very thick] (0.4,-1.5)--(2.6,-1.5);
\draw[-,  very thick] (0.4,-1.5)--(2.6,-3);
\draw[-,  very thick] (0.4,-3)--(2.6,-3);
\draw[-,  blue, very thick] (0.4,0)--(2.6,-1.5);
\end{tikzpicture}
\caption{\tbf{After} decomposition.}  \label{fig:dec/b}
\end{subfigure}

\caption{An example showing that for any target offline node, its WS structure must be instantiated when any of its offline neighbors shares a single online neighbor with it; otherwise, we can decompose it into another strictly worse structure for it. In the example above, the target offline node is $\bfi_1$, and it shares two online neighbors, $j_1$ and $j_2$, with its offline neighbor $i_2$. Note that all edges have a mass of $1/3$. The decomposition consists of (1) pruning the edge $(i_2, j_1)$ (marked in blue) and (2) adding another offline neighbor $\ti_2$ of $\bfi_1$ and edge $(\ti_2, j_1)$ to compensate for the role played by the pruned edge $(i_2,j_1)$. We can verify that at any time $t \in [0,1]$: (1) Given $i_1$ is safe at $t$, the probability of $\ti_2$ staying safe at $t$ and that of $\hi_2$ are both larger than that of $i_2$; (2) The real-time boosting effect for matching $\bfi_1$ from the unavailability of $i_2$ is equivalent to that from the unavailability of both $\ti_2$ and $\hi_2$. Thus, we conclude that $\bfi_1$ has a strictly worse structure in terms of a larger probability of staying safe after decomposition.}
\label{fig:dec}
\end{figure}

\section{When Mass-One Nodes Have Offline Neighbors of Mass Less Than One}\label{sec:mod}

In the case when all offline nodes have a unit mass after rounding, no modifications are required to any edges. This is due to the high symmetry present in the WS structures of an offline node, as depicted in Figures~\ref{fig:three/a} and~\ref{fig:three/b}. However, this symmetry breaks when a mass-one offline node neighbors another offline node of either mass 1/3 or 2/3. The instance in Figure~\ref{fig:29/a} highlights the necessity of modifying the sampling distributions of online nodes. For detailed discussions, see Appendix~\ref{app:mot/one}.


\subsection{Modifications to Sampling Distributions of Online Neighbors of Mass-One Nodes} \label{sec:tworep}

Consider a given offline node $\bfi$ of mass one. Following the \tbf{Worst-Scenario (WS) Principle}, its worst-case scenarios arise when each of its offline neighbors shares exactly one online neighbor with $\bfi$. As illustrated in Figure~\ref{fig:6/29}, there are three possible structures for an online neighbor $j$ of $\bfi$. For simplicity, we assume throughout this paper that every online node has a mass of one after rounding.

There are four possible combinations for a mass-one node $\bfi$:
\begin{itemize}
    \item Figure~\ref{fig:6/29/a} + Figure~\ref{fig:6/29/b}: Two online neighbors, one with the structure shown in Figure~\ref{fig:6/29/a} and the other with the structure shown in Figure~\ref{fig:6/29/b}.
    \item Figure~\ref{fig:6/29/a} + Figure~\ref{fig:6/29/c}.
    \item $3 \times$ Figure~\ref{fig:6/29/b}: Three online neighbors, each with the structure shown in Figure~\ref{fig:6/29/b}.
    \item $3 \times$ Figure~\ref{fig:6/29/c}: Three online neighbors, each with the structure shown in Figure~\ref{fig:6/29/c}.
\end{itemize}

Note that the example in Figure~\ref{fig:29/a} corresponds exactly to the case of Figure~\ref{fig:6/29/a} + Figure~\ref{fig:6/29/c}.

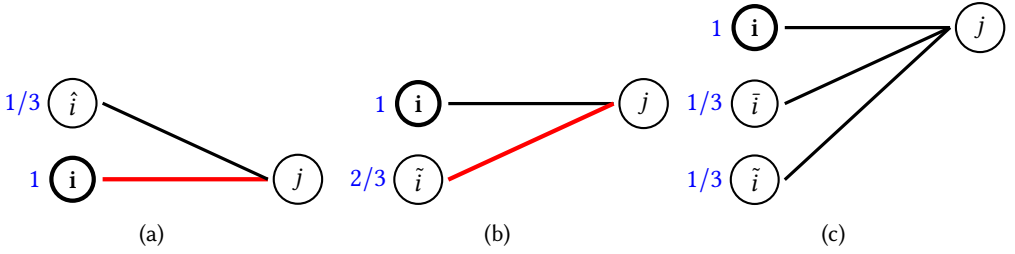
\begin{figure}[ht!]
\begin{subfigure}[b]{0.3\textwidth}
\begin{tikzpicture}
 \draw (0,0) node[minimum size=0.2mm,draw,circle, ultra thick] {$\mathbf{i}$};
  \draw (-0.3,0) node[left] {\bluee{$1$}};
    \draw (-0.3,1) node[left] {\bluee{$1/3$}};
 \draw (3,0) node[minimum size=0.2mm,draw,circle, thick] {$j$};
 \draw (0,1) node[minimum size=0.2mm,draw,circle, thick] {$\hi$};

\draw[-,  red, ultra thick] (0.4,0)--(2.6,0);
\draw[-,  very thick] (0.4, 1)--(2.6,0);
\end{tikzpicture}
\caption{}
    \label{fig:6/29/a}
\end{subfigure}
\hspace{0.1in}
\begin{subfigure}[b]{0.3\textwidth}
\begin{tikzpicture}
 \draw (0,0) node[minimum size=0.2mm,draw,circle, ultra thick] {$\mathbf{i}$};
  \draw (-0.3,0) node[left] {\bluee{$1$}};
    \draw (-0.3,-1) node[left] {\bluee{$2/3$}};
 \draw (3,0) node[minimum size=0.2mm,draw,circle, thick] {$j$};
 \draw (0,-1) node[minimum size=0.2mm,draw,circle, thick] {$\ti$};
\draw[-,  very thick] (0.4,0)--(2.6,0);
%
\draw[-,  red, ultra thick] (0.4, -1)--(2.6,0);
\end{tikzpicture}
\caption{}
    \label{fig:6/29/b}
\end{subfigure}
\hspace{0.05in}
\begin{subfigure}[b]{0.3\textwidth}   
\begin{tikzpicture}
 \draw (0,0) node[minimum size=0.2mm,draw,circle, ultra thick] {$\mathbf{i}$};
  \draw (-0.3,0) node[left] {\bluee{$1$}};
    \draw (-0.3,-1) node[left] {\bluee{$1/3$}};
    \draw (-0.3,-2) node[left] {\bluee{$1/3$}};
 \draw (3,0) node[minimum size=0.2mm,draw,circle, thick] {$j$};
 \draw (0,-1) node[minimum size=0.2mm,draw,circle, thick] {$\bi$};
  \draw (0,-2) node[minimum size=0.2mm,draw,circle, thick] {$\ti$};

\draw[-,  very thick] (0.4,0)--(2.6,0);
\draw[-,  very thick] (0.4, -1)--(2.6,0);
\draw[-,  very thick] (0.4, -2)--(2.6,0);

\end{tikzpicture}
\caption{}
    \label{fig:6/29/c}
\end{subfigure}
  \caption{Three possible structures of an \emph{online} neighbor for a mass-one offline node $\bfi$ under the \tbf{Worst-Scenario Principle}.}
      \label{fig:6/29}
\end{figure}

We consider the offline node $\bfi$ (mass one) in Figure~\ref{fig:6/29}. Its performance is expected to be worse than when all its offline neighbors have mass one. In other words, $\bfi$ is more likely to stay safe when its offline neighbors have a mass \anhai{of} less than one. This is because, at any given time, safe neighbors (like $\hi$, $\bi$, and $\ti$) with a lower mass have a higher chance of staying safe compared to those with mass one. This difference hinders the real-time boosting effect of online neighbor $j$ on matching $\bfi$. Conversely, as explained in Section~\ref{sec:ws1/3}, $\hi$, $\bi$, and $\ti$ perform better than their worst-case scenarios (when their offline neighbors have mass less than one). Motivated by this observation, we aim to identify appropriate modifications to the input vector on the online neighbor $j$ to achieve the following two goals:

\xhdr{Goal 1}: \emph{Any offline node of mass one in the form of 1B1S (one big edge, one small edge) or 3S (three small edges) should achieve an MPM no less than the target thresholds $\kapc$ and $\kapd$ in~\eqref{eqn:kap}, respectively}.

\xhdr{Goal 2}: \emph{Any offline nodes of mass 1/3 and 2/3 should achieve an MPM no less than the target thresholds $\kapa$ and $\kapb$ in~\eqref{eqn:kap}, respectively}.

\smallskip
It is important to note that for a given instance, there might be multiple choices of modifications that fulfill the two goals above. We devote the next section to the case when a mass-one offline node $\bfi$ has combined structures of Figures~\ref{fig:6/29/a} and~\ref{fig:6/29/c}. In other words, $\bfi$ has two online neighbors, one with the structure of Figure~\ref{fig:6/29/a} and the other with the structure of Figure~\ref{fig:6/29/c}. \emph{We defer the analysis of the remaining three combinations to Appendix~\ref{sec:7/7/a}}.

\subsection{Analysis of a Mass-One Node with Combined Structures:~\ref{fig:6/29/a}+\ref{fig:6/29/c}}\label{sec:a+c}

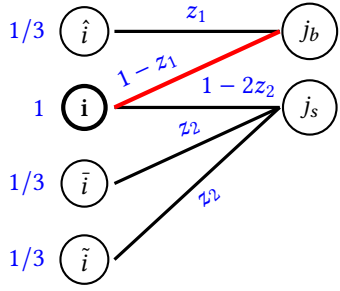
\begin{figure}[ht!]
\begin{tikzpicture}
 \draw (0,0) node[minimum size=0.2mm,draw,circle, ultra thick] {$\mathbf{i}$};
   \draw (-0.4,1) node[left] {\bluee{$1/3$}};
  \draw (-0.4,0) node[left] {\bluee{$1$}};
    \draw (-0.4,-1) node[left] {\bluee{$1/3$}};
    \draw (-0.4,-2) node[left] {\bluee{$1/3$}};
 \draw (3,0) node[minimum size=0.2mm,draw,circle, thick] {$j_s$};
  \draw (3,1) node[minimum size=0.2mm,draw,circle, thick] {$j_b$};
 \draw (0,-1) node[minimum size=0.2mm,draw,circle, thick] {$\bi$};
  \draw (0,-2) node[minimum size=0.2mm,draw,circle, thick] {$\ti$};
    \draw (0,1) node[minimum size=0.2mm,draw,circle, thick] {${\hi}$};

\draw[-,  very thick] (0.4,1)--(2.6,1);
\draw[-,  very thick] (0.4,0)--(2.6,0) node [blue, above, near end, sloped] {$1-2z_2$};
\draw[-,  very thick] (0.4, -1)--(2.6,0) node [blue, above, midway, sloped] {$z_2$};
\draw[-,  very thick] (0.4, -2)--(2.6,0) node [blue, below, midway, sloped] {$z_2$};

  \draw (1.5,1) node[above] {\bluee{$z_1$}};
\draw[-,  red, ultra thick] (0.4, 0)--(2.6,1) node [blue, above, near start, sloped] {$1-z_1$};   
\end{tikzpicture}
\caption{The offline node $\bfi$ (mass 1) has two online neighbors: one with the structure in Figure~\ref{fig:6/29/a} and the other in Figure~\ref{fig:6/29/c}. The values on the edges represent the updates proposed to achieve \tbf{Goals 1} and \tbf{2} in Section~\ref{sec:tworep}.}
    \label{fig:6/30/1}
    \end{figure}

Note that the original case corresponds to when $z_1 = z_2 = 1/3$. The lemma below suggests that the simple setting of $z_1 = z_2 = 0$ suffices to achieve both \tbf{Goal 1} and \tbf{Goal 2}.

\begin{lemma}
The configuration of $z_1 = z_2 = 0$ in Figure~\ref{fig:6/30/1} guarantees achieving both \tbf{Goal 1} and \tbf{Goal 2}, as stated in Section~\ref{sec:tworep}.
\end{lemma}

\begin{proof}
Focus on the setting of $z_1 = z_2 = 0$.  
Let $\alp(t)$ be the probability that the node $\bfi$ is safe at time $t$. Thus, $\alp(t) = \sfe^{-2t}$ and $1 - \alp(1) = 1 - \sfe^{-2} > \kap_B^*$.

Now, focus on analyzing $\hi$. Let $\alp(t)$ be the probability that $\hi$ is safe at time $t$ and $\beta(t)$ be the probability that both $\hi$ and $\bfi$ are safe at time $t$. Thus,
\begin{align*}
& \alp_{k+1} = \alp_k \cdot \bp{1 - \frac{1}{T} \bb{\frac{\beta}{\alp} \cdot 0 + \bp{1 - \frac{\beta}{\alp}} \cdot 1}}, \\
& \alp' + \alp - \beta = 0, \quad \alp(0) = 1.
\end{align*}
We see that $\alp(t) = 2\sfe^{-t} - \sfe^{-2t}$. Thus, we can verify that $(1 - \alp(1)) / (1/3) > 1$.  
Particularly, $\hi$ is matched in the end (at $t = 1$) with probability equal to
\[
1 - \alp(1) = 1 - 2\sfe^{-1} + \sfe^{-2} \approx 0.3996.
\]

Now, we analyze $\bi$ and $\ti$. We introduce a continuous-time Markov chain (MC) to compute the exact matching probabilities of $\bfi$, $\bi$, and $\ti$ as follows. There are in total four states, namely,
\[
s_1 = (1, 2), \quad s_2 = (0, 2), \quad s_3 = (0, 1), \quad s_4 = (0, 0),
\]
where in each state, the first value represents the status of $\bfi$ (1 means safe), while the second is the sum of safe nodes among $\bi$ and $\ti$. Note that states $(1, 1)$ and $(1, 0)$ both exist with probability zero due to the current specific sampling distribution on $j_s$: $\bi$ or $\ti$ can be matched by $j_s$ only after $\bfi$ becomes unavailable. For each $\ell = 1, 2, 3, 4$, let $q_\ell(t)$ denote the probability that the system is in state $\ell$ at time $t \in [0, 1]$ with initial conditions $q_1(0) = 1$ and $q_2(0) = q_3(0) = q_4(0) = 0$.  

\begin{table}[ht!]
\caption{The transition-rate matrix of the continuous-time Markov chain for Figure~\ref{fig:6/30/1} with $z_1 = z_2 = 0$.}
\label{table:12/24/a}
\begin{tabular}{c " cccc}
         & $(1, 2)$ & $(0, 2)$ & $(0, 1)$ & $(0, 0)$ \\ \thickhline 
$(1, 2)$ & $-2$     & $2$      & $0$      & $0$      \\
$(0, 2)$ & $0$      & $-1$     & $1$      & $0$      \\
$(0, 1)$ & $0$      & $0$      & $-1$     & $1$      \\
$(0, 0)$ & $0$      & $0$      & $0$      & $0$  \\   \thickhline 
\end{tabular}
\end{table}

Table~\ref{table:12/24/a} shows the transition-rate matrix of the continuous-time Markov chain. The Kolmogorov forward equations for the process are stated below:
\begin{align*}
&q'_1 = -2q_1, \quad q'_2 = -q_2 + 2q_1, \quad q'_3 = -q_3 + q_2, \quad q'_4 = q_3, \\
&q_1(0) = 1, \quad q_2(0) = q_3(0) = q_4(0) = 0. 
\end{align*}
We can solve that
\begin{align*}
q_1(t) &= \sfe^{-2t},  &&q_2(t) = 2\sfe^{-2t}(\sfe^t - 1), \\
q_3(t) &= 2\sfe^{-2t}(1 - \sfe^t + \sfe^t t), &&q_4(t) = \sfe^{-2t}(-1 + \sfe^{2t} - 2\sfe^t t).
\end{align*}
Thus, by symmetry, we see that each of $\bi$ and $\ti$ gets matched in the end with probability equal to
\[
q_4(1) + q_3(1) / 2 = 1 - 2 / \sfe \approx 0.2642,
\]
which implies that each achieves an MPM equal to $3(1 - 2 / \sfe) \approx 0.7927 > \kap_s^*$.  
\end{proof}

\section{Worst-Scenario Structures for Offline Nodes of Mass 1/3 and 2/3} \label{sec:ws1/3}

In this section, we aim to prove Claim (2) in Proposition~\ref{pro:main-1}, which states that any offline node of mass 1/3 and 2/3 achieves an MPM of at least $\kapa$ and $\kapb$ in \bt, respectively, where $\kapa$ and $\kapb$ are defined in~\eqref{eqn:kap}. According to the \tbf{Worst-Scenario (WS) Principle}, we can pinpoint the WS structures of an online neighbor for any offline node, as shown in Figure~\ref{fig:wsg}.
\begin{figure}[ht!]
\begin{subfigure}[b]{0.25\textwidth}
        \centering
        \resizebox{\linewidth}{!}{
\begin{tikzpicture}
  \draw (0,-1.5) node[minimum size=0.2mm,draw,circle, ultra thick] {$\bfi$};
    \draw (3,-1.5) node[minimum size=0.2mm,draw,circle, thick] {$j$};
    \draw (0,0) node[minimum size=0.2mm,draw,circle, thick] {{$\ti$}};
\draw[-,  red, ultra thick] (0.35, 0)--(2.6,-1.5);
\draw[-,   very thick] (0.35,-1.5)--(2.6,-1.5);
\end{tikzpicture}}
\caption{}  \label{fig:wsg/3}
\end{subfigure}
\begin{subfigure}[b]{0.25\textwidth}
        \centering
        \resizebox{\linewidth}{!}{
\begin{tikzpicture}
  \draw (0,-1.5) node[minimum size=0.2mm,draw,circle, ultra thick] {$\bfi$};
    \draw (0,1.5) node[minimum size=0.2mm,draw,circle, thick] {$\hi$};
    \draw (3,-1.5) node[minimum size=0.2mm,draw,circle, thick] {$j$};
    \draw (0,0) node[minimum size=0.2mm,draw,circle, thick] {{$\ti$}};
\draw[-,  very thick] (0.35, 0)--(2.6,-1.5);
\draw[-,   very thick] (0.35,-1.5)--(2.6,-1.5);
\draw[-,   very thick] (0.35,1.5)--(2.6,-1.5);
\end{tikzpicture}}
\caption{}  \label{fig:wsg/4}
\end{subfigure}
\begin{subfigure}[b]{0.25\textwidth}
        \centering
        \resizebox{\linewidth}{!}{
\begin{tikzpicture}
  \draw (0,-1.5) node[minimum size=0.3mm,draw,circle, ultra thick] {$\bfi$};
    \draw (3,-1.5) node[minimum size=0.2mm,draw,circle, thick] {$j$};
    \draw (0,0) node[minimum size=0.2mm,draw,circle, thick] {{$\ti$}};
\draw[-,  very thick] (0.35, 0)--(2.6,-1.5);
\draw[-,  red, ultra thick] (0.35,-1.5)--(2.6,-1.5);
\end{tikzpicture}}
\caption{}
 \label{fig:wsg/1}
\end{subfigure}
 \caption{All possible WS structures of an online neighbor with respect to the target offline node $\bfi$, according to the \tbf{Worst-Scenario (WS) Principle}. In the example above, small edges (with mass 1/3) and big edges (with mass 2/3) are marked in black and red, respectively. We list the different possible WS structures of an online neighbor $j$ of $\bfi$: Figure~\ref{fig:wsg/1} shows the case when $(\bfi, j)$ is a big edge, and Figures~\ref{fig:wsg/3} and~\ref{fig:wsg/4} show the cases when $(\bfi, j)$ is a small edge.}
    \label{fig:wsg}
\end{figure}
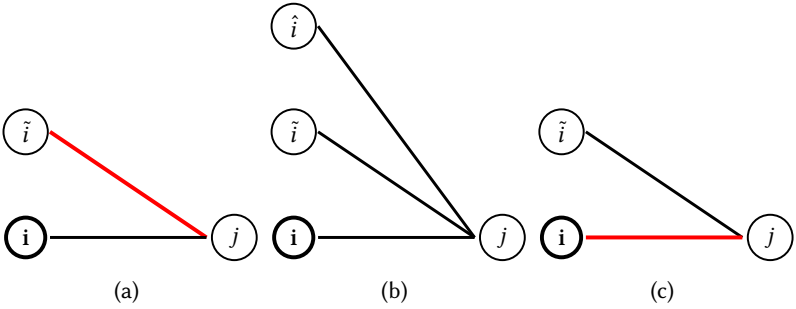

\smallskip

For an offline node $\bfi$ of mass 1/3, it has a single online neighbor. Thus, the WS structures should be instantiated as one in Figure~\ref{fig:wsg/3} or~\ref{fig:wsg/4}. For an offline node $\bfi$ of mass 2/3, the WS structures can be instantiated as either having a single online neighbor as shown in Figure~\ref{fig:wsg/1} or two online neighbors, each having a structure in Figures~\ref{fig:wsg/3} or~\ref{fig:wsg/4}. By the~\re, when an offline node of mass 2/3 has two online neighbors, they should share the same WS structure.

We split the discussion into the following cases: (\tbf{Case 1}) When a target node of mass 1/3 has a structure as in Figure~\ref{fig:wsg/3}; (\tbf{Case 2}) When a target node of mass 1/3 has a structure as in Figure~\ref{fig:wsg/4}; (\tbf{Case 3}) When a target node of mass 2/3 has two online neighbors, each with a structure as in Figure~\ref{fig:wsg/4}.\footnote{We omit the case when a target node $\bfi$ of mass 2/3 has a single online neighbor as in Figure~\ref{fig:wsg/1} because its performance is identical to that of node $\ti$ in Figure~\ref{fig:wsg/3}. Similarly, we skip the scenario when the target node $\bfi$ of mass 2/3 has two online neighbors, each with a structure as in Figure~\ref{fig:wsg/3}: Lemma~\ref{lem:2/3} in Appendix~\ref{app:ws/12/3} shows that node $\bfi$ performs strictly better than the case where each online neighbor has the structure as in Figure~\ref{fig:wsg/4}.} We summarize the numerical results in Table~\ref{table:sum123} and defer the analysis and justifications to Appendix~\ref{app:ws/12/3}.

\begin{table}[ht!]
\caption{Summary of \textbf{Matching Probability per Mass (MPM)} for an offline node of mass 1/3 or 2/3 in different WS structures outlined in Section~\ref{sec:ws1/3}. Note that \tbf{Case 3} considers one possible WS structure for an offline node of mass 2/3 when it has two online neighbors, each with a structure as in Figure~\ref{fig:wsg/4}. All fractional values are rounded to the fourth decimal \anhai{place} if needed. The values marked in blue are those obtained after applying modifications to the sampling distributions on the online node $j$; see details in Appendix~\ref{app:ws/12/3}.}
\label{table:sum123}
\begin{tabular}{c " c ccccccc} \thickhline 
   &  \tbf{Case 1} (Figure~\ref{fig:wsg/3})  & \tbf{Case 2} (Figure~\ref{fig:wsg/4}) & \tbf{Case 3} (Figure~\ref{fig:wsg/4}+Figure~\ref{fig:wsg/4})   \\ \thickhline 
(1/3) &   $1.1606 \to \bluee{0.8963}$ & {0.9766}  & $\ge 0.9766$\\
(2/3) &    {$0.7642 \to \bluee{0.8963}$} &  & 0.8177  \\ \thickhline 
\end{tabular}
\end{table}

\xhdr{Remarks on results in Table~\ref{table:sum123}}. As indicated by the results in Table~\ref{table:sum123}, the worst-scenario structure for an offline node of mass 2/3 corresponds to Figure~\ref{fig:wsg/3} (before any modifications), achieving an MPM of $0.7642$, which is less than the target $\kapb$. \emph{In Appendix~\ref{app:2/3}, we demonstrate that both $\bfi$ and $\ti$ can achieve an MPM of $2 - 3/\sfe \approx 0.8963$, surpassing $\kapb$, after adjusting the sampling distribution on $j$ as follows}: If both $\bfi$ and $\ti$ are safe upon $j$'s arrival, match $j$ to $\bfi$ and $\ti$ with respective probabilities $z$ and $1 - z$, where $z = 1 - \sfe/3$; if only one is safe, match $j$ with that safe node with probability one.

\section{Conclusion and Future Directions}
In this paper, we examined vertex-weighted online matching under KIID with integral arrival rates. We introduced a meta algorithm (\bt) that employs real-time boosting. To showcase $\bt$'s efficacy, we demonstrated that $\bt(\X^*)$ achieves a competitiveness of at least 0.7341, surpassing the current state-of-the-art results of 0.7299~\cite{brubach2020online} and 0.725~\cite{bib:Jaillet}, where $\X^*$ is a random vector obtained using $\mathsf{DR}[3]$ from~\cite{brubach2020online}. Concurrently, we proposed an auxiliary algorithm (\vir) that highlights the subtle connections between the algorithms in~\cite{brubach2020online},~\cite{bib:Jaillet}, and our proposed methodology.

\smallskip

Our work opens several future directions. The first is to apply the ODEs system-based competitive analysis approach presented in this paper to more general settings, such as edge-weighted matching and/or general arrival rates. The second is to sharpen the upper bound and develop an improved hardness result specifically for the setting considered here. So far, the best hardness result for online matching under KIID with integral arrival rates is $1 - \sfe^{-2} \approx 0.8646$ due to~\cite{bib:Manshadi}, which is based on an \emph{unweighted} instance. Can we achieve a tighter bound by considering a vertex-weighted case?

\newpage
\bibliographystyle{alpha}
\bibliography{EC_21}

\newcommand{\etalchar}[1]{$^{#1}$}
\begin{thebibliography}{FMMM09}

\bibitem[BSSX20]{brubach2020online}
Brian Brubach, Karthik~Abinav Sankararaman, Aravind Srinivasan, and Pan Xu.
\newblock Online stochastic matching: New algorithms and bounds.
\newblock {\em Algorithmica}, 82(10):2737--2783, 2020.

\bibitem[DSS{\etalchar{+}}19]{aamas-19}
John~P Dickerson, Karthik~Abinav Sankararaman, Kanthi~Kiran Sarpatwar, Aravind
  Srinivasan, Kun-Lung Wu, and Pan Xu.
\newblock Online resource allocation with matching constraints.
\newblock In {\em Proceedings of the 18th International Conference on
  Autonomous Agents and MultiAgent Systems}, pages 1681--1689. International
  Foundation for Autonomous Agents and Multiagent Systems, 2019.

\bibitem[DSSX18]{dickerson2018assigning}
John~P Dickerson, Karthik~Abinav Sankararaman, Aravind Srinivasan, and Pan Xu.
\newblock Assigning tasks to workers based on historical data: Online task
  assignment with two-sided arrivals.
\newblock In {\em Proceedings of the 17th International Conference on
  Autonomous Agents and MultiAgent Systems}, pages 318--326. International
  Foundation for Autonomous Agents and Multiagent Systems, 2018.

\bibitem[DSSX21]{teac21}
John~P Dickerson, Karthik~A Sankararaman, Aravind Srinivasan, and Pan Xu.
\newblock Allocation problems in ride-sharing platforms: Online matching with
  offline reusable resources.
\newblock {\em ACM Transactions on Economics and Computation (TEAC)},
  9(3):1--17, 2021.

\bibitem[FMMM09]{bib:Feldman}
Jon Feldman, Aranyak Mehta, Vahab Mirrokni, and S~Muthukrishnan.
\newblock Online stochastic matching: Beating 1-1/e.
\newblock In {\em Foundations of Computer Science, 2009. FOCS'09. 50th Annual
  IEEE Symposium on}, pages 117--126. IEEE, 2009.

\bibitem[FMSL19]{fata2019multi}
Elaheh Fata, Will Ma, and David Simchi-Levi.
\newblock Multi-stage and multi-customer assortment optimization with inventory
  constraints.
\newblock {\em Available at SSRN 3443109}, 2019.

\bibitem[FNS19]{feng2019linear}
Yiding Feng, Rad Niazadeh, and Amin Saberi.
\newblock Linear programming based online policies for real-time assortment of
  reusable resources.
\newblock {\em Chicago Booth Research Paper}, (20-25), 2019.

\bibitem[GGI{\etalchar{+}}21]{gong2021online}
Xiao-Yue Gong, Vineet Goyal, Garud~N Iyengar, David Simchi-Levi, Rajan Udwani,
  and Shuangyu Wang.
\newblock Online assortment optimization with reusable resources.
\newblock {\em Management Science}, 2021.

\bibitem[GKPS06]{gandhi2006dependent}
Rajiv Gandhi, Samir Khuller, Srinivasan Parthasarathy, and Aravind Srinivasan.
\newblock Dependent rounding and its applications to approximation algorithms.
\newblock {\em Journal of the ACM (JACM)}, 53(3):324--360, 2006.

\bibitem[HMZ11]{bib:Haeupler}
Bernhard Haeupler, Vahab~S. Mirrokni, and Morteza Zadimoghaddam.
\newblock Online stochastic weighted matching: Improved approximation
  algorithms.
\newblock In {\em Internet and Network Economics}, volume 7090 of {\em Lecture
  Notes in Computer Science}, pages 170--181. Springer Berlin Heidelberg, 2011.

\bibitem[HS21]{huang2021online}
Zhiyi Huang and Xinkai Shu.
\newblock Online stochastic matching, poisson arrivals, and the natural linear
  program.
\newblock In {\em Proceedings of the 53rd Annual ACM SIGACT Symposium on Theory
  of Computing}, pages 682--693, 2021.

\bibitem[HSY22]{huang2022power}
Zhiyi Huang, Xinkai Shu, and Shuyi Yan.
\newblock The power of multiple choices in online stochastic matching.
\newblock In {\em Proceedings of the 54th Annual ACM SIGACT Symposium on Theory
  of Computing}, pages 91--103, 2022.

\bibitem[HV12]{ho2012online}
Chien-Ju Ho and Jennifer~Wortman Vaughan.
\newblock Online task assignment in crowdsourcing markets.
\newblock In {\em Twenty-sixth AAAI conference on artificial intelligence},
  2012.

\bibitem[JL13]{bib:Jaillet}
Patrick Jaillet and Xin Lu.
\newblock Online stochastic matching: New algorithms with better bounds.
\newblock {\em Mathematics of Operations Research}, 39(3):624--646, 2013.

\bibitem[KVV90]{kvv}
Richard~M. Karp, Umesh~V. Vazirani, and Vijay~V. Vazirani.
\newblock An optimal algorithm for on-line bipartite matching.
\newblock In {\em Proceedings of the 22nd Annual {ACM} Symposium on Theory of
  Computing}, STOC '90, pages 352--358, 1990.

\bibitem[Meh13]{mehta2012online}
Aranyak Mehta.
\newblock Online matching and ad allocation.
\newblock {\em Foundations and Trends in Theoretical Computer Science},
  8(4):265--368, 2013.

\bibitem[MGS12]{bib:Manshadi}
Vahideh~H Manshadi, Shayan~Oveis Gharan, and Amin Saberi.
\newblock Online stochastic matching: Online actions based on offline
  statistics.
\newblock {\em Mathematics of Operations Research}, 37(4):559--573, 2012.

\bibitem[MXX21]{ma2021grouplevel}
Will Ma, Pan Xu, and Yifan Xu.
\newblock Group-level fairness maximization in online bipartite matching.
\newblock {\em arXiv preprint arXiv:2011.13908}, 2021.

\bibitem[MXX23]{ma2023fairness}
Will Ma, Pan Xu, and Yifan Xu.
\newblock Fairness maximization among offline agents in online-matching
  markets.
\newblock {\em ACM Transactions on Economics and Computation}, 10(4):1--27,
  2023.

\bibitem[QFZW23]{qiu2023improved}
Guoliang Qiu, Yilong Feng, Shengwei Zhou, and Xiaowei Wu.
\newblock Improved competitive ratio for edge-weighted online stochastic
  matching.
\newblock In {\em International Conference on Web and Internet Economics},
  pages 527--544. Springer, 2023.

\bibitem[TWW22]{tang2022fractional}
Zhihao~Gavin Tang, Jinzhao Wu, and Hongxun Wu.
\newblock (fractional) online stochastic matching via fine-grained offline
  statistics.
\newblock In {\em Proceedings of the 54th Annual ACM SIGACT Symposium on Theory
  of Computing}, pages 77--90, 2022.

\bibitem[Yan24]{yan2024edge}
Shuyi Yan.
\newblock Edge-weighted online stochastic matching: Beating.
\newblock In {\em Proceedings of the 2024 Annual ACM-SIAM Symposium on Discrete
  Algorithms (SODA)}, pages 4631--4640. SIAM, 2024.

\bibitem[ZXS{\etalchar{+}}19]{xu-aaai-19}
Boming Zhao, Pan Xu, Yexuan Shi, Yongxin Tong, Zimu Zhou, and Yuxiang Zeng.
\newblock Preference-aware task assignment in on-demand taxi dispatching: An
  online stable matching approach.
\newblock In {\em Proceedings of the Thirty-Third Conference on Artificial
  Intelligence}, AAAI '19, pages 2245--2252, 2019.

\end{thebibliography}
 
\clearpage
\appendix

\section{Proof of Lemma~\ref{lem:well-bt}} \label{app:well-bt} 
\begin{figure}[ht!]
\begin{tikzpicture}
 \draw (0,0) node[minimum size=0.2mm,draw,circle, ultra thick] {$\mathbf{i}$};

  \draw (-0.4,0) node[left] {\bluee{$1-\sfe^{-K}$}};
    \draw (-0.4,-1) node[left] {\bluee{$\ep_k$}};
    \draw (-0.4,-2) node[left] {\bluee{$\ep_k$}};
        \draw (-0.4,-4) node[left] {\bluee{$\ep_k$}};
  \draw (3,0) node[minimum size=0.2mm,draw,circle, thick] {$j_k$};
 \draw (0,-1) node[minimum size=0.2mm,draw,circle, thick] {$i_1$};
  \draw (0,-2) node[minimum size=0.2mm,draw,circle, thick] {$i_2$};
    \draw (0,-4) node[minimum size=0.2mm,draw,circle, thick] {$i_N$};
\draw[-,  very thick] (0.4,0)--(2.6,0);
\draw[-,  very thick] (0.4, -1)--(2.6,0);
\draw[-, very thick] (0.4, -2)--(2.6,0);
\draw[-,  thick, dashed] (0.4, 0)--(2.6,1);
\draw[-,  thick, dashed] (0.4, 0)--(2.6,2);
\draw[-,  thick, dashed] (0.4, -3)--(2.6,0);
\draw[-,  very thick] (0.4, -3.8)--(2.6,0);

          \draw (1.6,-0.5) node[above] {\bluee{$\ep_k$}};
              \draw (1.6,-0.05) node[above] {\bluee{$\del_k$}};
\end{tikzpicture}
\caption{An instance where $\bt(\x^*)$ achieves a competitive ratio of no more than $1-1/\sfe$, where $\x^*$ is an optimal solution to the natural LP in~\cite{huang2021online}. In the instance above, $i$ has $K \gg 1$ online neighbors, namely, $j_1, j_2, \ldots, j_K$, such that each has the same structure as $j_k$, which has $N$ online neighbors other than $i$. For each $k \in [K] := \{1,2,\ldots,K\}$, set $\del_k = \sfe^{-(k-1)} - \sfe^{-k}$, $\ep_k = (1 - \del_k)/N$, and let the vertex weight on $i$ be one, dominating the sum of weights over all the other nodes. We verify that $\x^* = (x^*_{ij})$ is an optimal solution to the natural LP in~\cite{huang2021online}, where for each $k \in [K]$, $x^*_{i, j_k} = \del_k$ and $x^*_{i', j_k} = \ep_k$ for all $i' \neq i, i' \in I_{j_k}$.}\label{fig:1/1}
\end{figure}
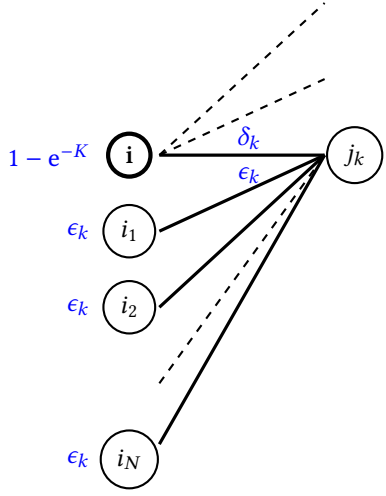

\begin{proof}
Consider the instance shown in Figure~\ref{fig:1/1}. We verify that $\x^* = (x^*_{ij})$ is an optimal solution to the natural LP in~\cite{huang2021online}, where for each $k \in [K]$, $x^*_{i, j_k} = \del_k$ and $x^*_{i', j_k} = \ep_k$ for all $i' \in I_{j_k} \setminus \{i\}$.\footnote{Note that polynomial-time algorithms for LP, such as the interior point method, do not always yield an extreme-point optimal solution. Therefore, in the competitive analysis for $\bt(\x^*)$, we cannot assume conclusively that $\x^*$ is an extreme point of any benchmark LP, and as a result, we cannot exploit related properties. This is why, despite $\x^*$ not being an extreme point of the natural LP polytope in~\cite{huang2021online}, we can confidently assert that $\x^*$ is an optimal solution to be utilized in feeding $\bt$.} 

\smallskip

Now, we aim to show that for each $k \in [K]$ with $\del_k = \del \in [0,1]$, $i$ will survive from being matched by $j_k$ with a probability equal to $\sfe^{-\del}$ in $\bt(\x^*)$ when $N$ approaches infinity. This suggests that $i$ gets matched in $\bt(\x^*)$ with a probability equal to $1 - \sfe^{-\sum_{k \in [K]} \del_k}$ (as $N \to \infty$). Meanwhile, any offline optimal can match $i$ with a probability equal to $1 - \sfe^{-K}$. Thus, the resulting competitiveness is 
\[
\frac{1 - \sfe^{-\sum_{k \in [K]} \del_k}}{1 - \sfe^{-K}} = \frac{1 - \sfe^{-(1 - \sfe^{-K})}}{1 - \sfe^{-K}} \Rightarrow 1 - \frac{1}{\sfe}, \quad \text{as } K \to \infty.
\]

\smallskip

Consider a given $k$ with $j_k = j$, and a fixed value $\del_k = \del$ with $\ep_k = \ep = (1 - \del) / N$ (we drop the subscript $k$ when the context is clear). Observe that $i$ survives from being matched by $j$, denoted by $A = 1$, if and only if (1) $j$ arrives no more than $N$ times, and (2) for every arrival $\ell \in [N]$, $j$ happens to match some $i' \neq i$ following the boosted sampling distribution in $\bt(\x^*)$. Specifically, we have
\begingroup
\allowdisplaybreaks
\begin{align*}
\E[A] &= \sfe^{-1} + \sum_{\ell=1}^N \frac{\sfe^{-1}}{\ell!} (1 - \del) \cdot \frac{1 - \ep - \del}{1 - \ep} \cdots \frac{1 - (\ell - 1) \ep - \del}{1 - (\ell - 1) \ep} \\
&= \sfe^{-1} + \sum_{\ell=1}^N \frac{\sfe^{-1}}{\ell!} \prod_{\ell'=1}^\ell \frac{1 - (\ell' - 1) \ep - \del}{1 - (\ell' - 1) \ep}.
\end{align*}
\endgroup
We verify that $\E[A] \to \sfe^{-\del}$ as $N \to \infty$. Thus, we complete the proof.
\end{proof}


\section{Proof of Proposition~\ref{pro:main-0}}\label{app:main-0}
\begin{proof}
Consider a given optimal solution $\y=(y_{e})$ to the benchmark \LP-\eqref{lp}. Recall that the specialized dependent rounding in~\cite{brubach2020online}, denoted as $\mathsf{DR}[3]$, processes $\y$ as follows: First, apply the classical dependent rounding to $3 \y=(3 y_{e})$ and obtain a rounded integer vector $\Y=(Y_e) \in \{0,1,2\}^{|E|}$; second, set $\X=(X_e)$ with $X_e=Y_e/3$ for every $e \in E$. Consider a given offline node $i$ with $y_i:=\sum_{e \in E_i} y_{e}$, where $E_i$ denotes the set of edges incident to $i$ in the input graph $G$. Observe that $i$ has a big edge in $G(\X)$ if and only if there exists at least one edge $e \in E_i$ such that $y_e > 1/3$, since otherwise the fact that $3 y_e \le 1$ implies $Y_e \in \{0,1\}$ and $X_e \in \{0,1/3\}$. Meanwhile, $E_i$ has no more than two edges of mass larger than $1/3$ under $\y$ since $y_i=\sum_{e \in E_i} y_e \le 1$. Consider the following two cases.

\smallskip

\tbf{Case 1}. $E_i$ has only one edge $e$ with $y_e > 1/3$. By Constraint~\eqref{lp-a}, $y_e \le 1 - 1/\sfe$. Thus, $i$ has a big edge in $G(\X)$ if and only if $Y_e=2$ and $X_e=Y_e/3=2/3$, which happens with probability equal to $3y_e - 1 \le 2 - 3/\sfe$.

\smallskip

\tbf{Case 2}. $E_i$ has two edges, say $e$ and $e'$, such that $y_e > 1/3$ and $y_{e'} > 1/3$. By Constraint~\eqref{lp-b}, $y_e + y_{e'} \le 1 - \sfe^{-2}$. Thus, $i$ has a big edge in $G(\X)$ if either $Y_e=2$ or $Y_{e'}=2$, which occurs with probability equal to $3(y_e + y_{e'}) - 2 \le 1 - 3 \sfe^{-2} < 2 - 3/\sfe$.
\end{proof}

\section{Proof of Theorem~\ref{thm:main-1}}\label{app:main-thm}
\begin{proof}
Let $\y = (y_{ij})$ denote an optimal solution to the benchmark \LP-\eqref{lp}. Suppose $\X = (X_{ij}) \in \{0,1/3,2/3\}$ represents the random rounded vector obtained after applying the specialized dependent rounding $\mathsf{DR}[3]$ to $\y$. Consider a given offline node $i \in I$, and let $y_i := \sum_{j \in J_i} y_{ij} \in [0,1]$ and $X_i := \sum_{j \in J_i} X_{ij}$. By the marginal-distribution property, $\E[X_i] = y_i$. Consider the following two cases.

\tbf{Case 1}. $0 \le y_i \le 2/3$. We see that $3y_i \in [0,2]$ and $X_i \in \{0,1/3,2/3\}$. For each $\ell \in \{0,1/3,2/3\}$, let $p_\ell := \Pr[X_i = \ell]$, and let $q_\ell$ be the matching probability of $i$ in $\bt$ given $X_i = \ell$. Note that the overall matching probability of $i$ in $\bt$ is equal to
\[
q^{(i)} := \sum_{\ell \in \{0,1/3,2/3\}} p_\ell \cdot q_\ell \ge \sum_{\ell \in \{0,1/3,2/3\}} p_\ell \cdot (\ell \cdot \kap^*) = \kap^* \sum_{\ell \in \{0,1/3,2/3\}} p_\ell \cdot \ell = \kap^* \cdot y_i,
\]
where the inequality above follows from Proposition~\ref{pro:main-1} and the facts that $\kapa = \kap^*$ and $\kapb > \kap^*$.

\tbf{Case 2}. $2/3 \le y_i \le 1$. Then $3y_i \in [2,3]$ and $X_i \in \{2/3,1\}$. Let $p_a$, $p_b$, and $p_c$ denote the probability that $X_i = 2/3$, $X_i = 1$ and $i$ has one big and one small edge (1B1S), and $X_i = 1$ and $i$ has three small edges (3S), respectively. The overall matching probability of $i$ in $\bt$ is at least
\[
q^{(i)} := p_a \cdot (2/3) \cdot \kapb + p_b \cdot \kapc + p_c \cdot \kapd
\]
by Proposition~\ref{pro:main-1}. Consider the minimization program below:
\begin{align}
\min & ~\frac{q^{(i)}}{y_i} := \frac{p_a \cdot (2/3) \cdot \kapb + p_b \cdot \kapc + p_c \cdot \kapd}{p_a \cdot (2/3) + p_b + p_c} \label{min:q} \\
& p_a + p_b + p_c = 1, \nonumber \\
& p_b \le 2 - 3/\sfe, \nonumber \\
& 0 \le p_a, p_b, p_c \le 1. \nonumber
\end{align}
We can verify that Program~\eqref{min:q} achieves a minimum value of $\kap^*$, where the bottleneck case occurs at either $y_i = 1$ or $y_i = 4/3 - \sfe^{-1} \approx 0.9654$, and in both cases, $i$ is instantiated as 1B1S with probability equal to $2 - 3/\sfe$. Summarizing the above two cases, we claim that for any offline node $i$ with any value $y_i \in [0,1]$, the ratio of its matching probability in $\bt$ to its mass $y_i$ (under $\y$) is at least $\kap^*$.

As a result, the expected total weight achieved by $\bt$ is
\[
\sum_{i \in I} w_i \cdot q^{(i)} \ge \kap^* \sum_{i \in I} w_i \cdot y_i = \kap^* \cdot \mathsf{Val}(\LP\mbox{-}\eqref{lp}) \ge \kap^* \OPT,
\]
where $\mathsf{Val}(\LP\mbox{-}\eqref{lp})$ and $\OPT$ denote the optimal value of \LP-\eqref{lp} and the performance of an offline optimal policy, respectively, and where the last inequality follows from the fact that \LP-\eqref{lp} is a valid upper bound on the performance of any offline optimal policy.
\end{proof}

\section{Proof of Lemma~\ref{lem:vir}} \label{app:vir} 
\begin{proof}
Consider a given  $\X=\x \in \{0,1/3,2/3\}^{|E|}$ and {assume}  $x_j:=\sum_{i \in I_{j}} x_{ij}=1$ WLOG. Suppose an online node of type $j \in J$ arrives at some time $t \in [0,1]$. Focus on the case when $|I_{j,\x}|=3$ with $I_{j, \x}=\{i_1, i_2, i_3\}$. Let $I_{j,\x,t}=I_{j,t} \cap I_{j,\x}$ be the set of non-zero neighbors of $j$ that are safe at $t$. Observe that for $\bt(\x)$, Step~\eqref{alg:meta/s2} is equivalent to sampling a neighbor $i \in I_{j,\x,t}$ with probability $x_{ij}/\sum_{i' \in I_{j,\x,t}} x_{i',j}$ since we can ignore all zero neighbors $i$ with $x_{ij}=0$ and remove them from $I_{jt}$.

We prove Claim (1) first by splitting {into the} following scenarios. Let $\widetilde{\mathsf{AUG}}(\x)$ {refer} to the updated version of $\vir$ fed with $\x$.

\tbf{Case 1}. All the three non-zero neighbors are safe (unmatched) at $t$, \ie $I_{j,\x,t}=\{i_1, i_2, i_3\}$. We can verify that (1) each $i \in \{i_1, i_2, i_3\}$ gets sampled and matched in $\bt(\x)$ with probability equal to $x_{ij}$ (note that $\sum_{i \in I_{j,\x, t}} x_{ij}=\sum_{i \in I_{j,\x}} x_{ij}=1$); and (2) each $i \in \{i_1, i_2, i_3\}$ gets matched in $\widetilde{\mathsf{AUG}}(\x)$ iff $i$ tops the list $\cL_j$, which occurs with probability equal to $x_{ij}$. Thus, each $i \in \{i_1, i_2, i_3\}$ gets matched with the same probability in $\bt(\x)$ and $\widetilde{\mathsf{AUG}}(\x)$.

\smallskip

\tbf{Case 2}. There are only two non-zero neighbors that are safe at $t$, say,  $I_{j,\x,t}=\{i_1, i_2\}$. In this case, each $i \in \{i_1, i_2\}$ gets matched in $\bt(\x)$
 with probability equal to $x_{ij}/(x_{i_1,j}+x_{i_2,j})$. Meanwhile, for the node $i_1$, it gets matched in $\widetilde{\mathsf{AUG}}(\x)$ iff either $i_1$ tops the list $\cL_j$ with $\cL_j=(i_1, *,*)$, or $i_1$ is the second while $i_3$ is the first on $\cL_j$ with $\cL_j=(i_3,i_1,i_2)$. The former occurs with probability equal to $x_{i_1,j}$ while the latter equal to $x_{i_3,j} \cdot x_{i_1,j}/(x_{i_1,j}+x_{i_2,j})$. Thus, the total probability that $i_1$ gets matched in $\vir(\x)$ is equal to 
 \[
 x_{i_1,j}+x_{i_3,j} \frac{x_{i_1,j}}{x_{i_1,j}+x_{i_2,j}}=x_{i_1,j} \cdot \bp{1+\frac{x_{i_3,j}}{x_{i_1,j}+x_{i_2,j}}}=\frac{x_{i_1,j}}{x_{i_1,j}+x_{i_2,j}},
 \]
which is identical to the probability that  $i_1$ gets matched in $\bt(\x)$. We can argue similarly for $i_2$.

\smallskip

\tbf{Case 3}. There are only one non-zero {neighbor that is} safe at $t$, say,  $I_{j,\x,t}=\{i_1\}$. In this case, $i_1$ gets matched in $\bt(\x)$
 with probability one. Meanwhile, $i_1$ gets sampled and matched in $\widetilde{\mathsf{AUG}}(\x)$ with probability one as well since the other two non-zero neighbors are both matched then.

Now, we prove Claim (2). Since $\widetilde{\mathsf{AUG}}$ is identical to $\bt$, it suffices to show that the performance of $\widetilde{\mathsf{AUG}}$ is lower bounded by that of $\vir$. Consider a given realization path, which is captured by a specific sequence of outcomes of random lists generated for all arriving nodes. Focus on a given offline node $i$, and let $J_{i,\x} \subseteq J_i$ be the set of non-zero neighbors of $i$ under $\x$.  Observe that (1) $i$ is matched in $\widetilde{\mathsf{AUG}}$ iff there exists some $\cL_j$ with $j \in J_{i,\x}$ such that every node $i'$ prior to $i$ on $\cL_j$ get matched at the time when $j$ arrives; and (2) $i$ is matched in $\vir$ iff there exists some $\cL_j$ with $j \in J_{i,\x}$ such that every node $i'$ prior to $i$ on $\cL_j$ get matched \emph{by $j$ itself} at the time when $j$ arrives. By comparing these two conditions, we see that the latter is stricter than the former, and thus, any realization path leading to $i$'s match in $\vir$ can secure $i$' match in $\widetilde{\mathsf{AUG}}$, but not vice-versa. Thus, we establish Claim (2).
\end{proof}

     \section{Analysis of the Example in Figure~\ref{fig:12/27/a}}\label{app:mot/one}

   \begin{lemma}\label{lem:7/3/a}
The  Matching Probability per Mass (MPM) achieved by node $\bfi$ of Figure~\ref{fig:29/a} in \bt (or $\widetilde{\mathsf{AUG}}$) is equal to $1-22/(9 \sfe^2) \approx 0.6692<\kapc$,  the target MPM for an offline node of mass one in the form of 1B1S, as specified in~\eqref{eqn:kap}.
\end{lemma}
  
\begin{proof}
We introduce a  (discrete) Markov Chain (MC) and use it to conduct a holistic competitive analysis for the algorithm~\bt. There are in total seven states, namely,
\[
s_1=(1,1,2), s_2=(0,1,2), s_3=(1,1,1), s_4=(0,1,1), s_5=(1,1,0), s_6=(0,1,0), s_7=(*,0,*),
\]
where in each state, the first and the second values represent the status of $\hi$ and $\bfi$ (1 means safe), respectively, while the third one {is} the number of safe nodes between $\bi$ and $\ti$.  The MC starts at state $s_1$ with probability one, and the one-step transition matrix states below:
\[
\mathbf{H}=\begin{pmatrix}
1 - \frac{2}{T} & \frac{1/3}{T} &\frac{2/3}{T} & 0& 0& 0 & \frac{1}{T}\\
0 & 1-\frac{2}{T} & 0& \frac{2/3}{T} & 0 & 0 &\frac{4/3}{T}\\
0 & 0 & 1 - \frac{2}{T} & \frac{1/3}{T} & \frac{1/2}{T}& 0 & \frac{7/6}{T} \\
0 & 0 & 0& 1-\frac{2}{T}& 0&\frac{1/2}{T} &\frac{3/2}{T}\\
0 & 0 & 0& 0& 1-\frac{2}{T}&\frac{1/3}{T} & \frac{5/3}{T}\\
0 & 0 & 0& 0& 0& 1-\frac{2}{T} &\frac{2}{T}\\
0 & 0 & 0& 0& 0& 0&1
\end{pmatrix}.
\]  
By taking the limit of $\lim_{T \to \infty}\mathbf{H}^T$, we can get the final stationary distribution $\bpi=(\pi_\ell)$, where $\pi_\ell$ denotes the stationary probability of ending at state $s_\ell$ when starting at $s_1$. Specifically, we have that the target node $\bfi$ gets matched with {a} probability of $\pi_7=1-22/(9 \sfe^2) \approx 0.6692$, which suggests $\bfi$ achieves an MPM strictly less than $\kap^*_B =1-2 \sfe^{-2}$, the target MPM for an offline node of mass one in the form of 1B1S.
\end{proof}

     \begin{lemma}\label{lem:7/3/b}
The  Matching Probability per Mass (MPM) achieved by nodes $\bi$ and $\ti$ of Figure~\ref{fig:29/b} in $\vir$ (the exact algorithm analyzed in~\cite{bib:Jaillet, brubach2020online}) are both  equal to 
\[
\frac{3 \left(8 z_2^2+\frac{-14 z_2^2+19 z_2+9}{\sfe}-12 z_2-4\right)}{4 (z_2-1)}. 
\]
\end{lemma}
We introduce the following continuous Markov chain, which consists of six states
\[
s_1=(1,2), s_2=(1,1), s_3=(0,2), s_4=(1,0), s_5=(0,1), s_6=(0,0),
\]
where in each state, the first number is equal to $1$ indicates $\bfi$ is safe equal to $0$ otherwise, and the second represents the number of safe offline neighbors among $\bi$ and $\ti$. The transition-rate matrix is shown in Table~\ref{table:7/3/1}.

 \begin{table}[ht!]
\caption{The transition-rate matrix of the continuous-time Markov chain for Lemma~\ref{lem:7/3/b}.}
\label{table:7/3/1}
\begin{tabular}{c " c ccccc} \thickhline 
         & $(1,2)$ &$(1,1)$ &$(0,2)$ &$(1,0)$ &$(0,1)$ &$(0,0)$   \\ \thickhline 
(1,2)  & -1 &$2z$ &$1-2z$        & 0 &0 &0  \\ 
(1,1)  & 0 &-1 &0        & $z/(1-z)$ &$(1-2z)/(1-z)$ &0  \\ 
(0,2)  & 0 &0 &-1        & 0 &1 &0  \\ 
(1,0)  & 0 &0 &0        & -1 &0 &1  \\ 
(0,1)  & 0 &0 &0        & 0 &-1 &1  \\ 
(0,0)  & 0 &0 &0        & 0 &0 &0  
\end{tabular}
\end{table}
\begin{proof}
Let $q_\ell(t)$ denote the probability that the system falls at state $s_\ell$ at time $t \in [0,1]$ for $1\le \ell \le 6$.  The Kolmogorov forward equations for the process are \anhai{stated} below: 
\begingroup
\allowdisplaybreaks
\begin{align*}
q_1' &=- q_1,\\
q_2' &=2z q_1 -  q_2,\\
q_3' &=(1-2z) q_1 -  q_3,\\
q_4' &=\frac{z}{1-z} q_2 -  q_4,\\
q_5' &=\frac{1-2z}{1-z} q_2+q_3 -  q_5,\\
q_6' &=q_4+q_5,
\end{align*}
\endgroup
where the initial conditions are stated as $q_1(0)=1$, and $q_\ell(0)=0, \forall 2 \le \ell \le 6$. Thus, for offline nodes $\bi$ and $\ti$, they each achieve an MPM equal to
\[
\bp{\sbp{q_2(1)+q_5(1)}/2+q_4(1)+q_6(1)}/(1/3)=\frac{3 \left(8 z_2^2+\frac{-14 z_2^2+19 z_2+9}{\sfe}-12 z_2-4\right)}{4 (z_2-1)}.
\]
\end{proof}
We can verify that when $z_2 = 0$, the expression above simplifies to $3 \sbp{1-9/(4\sfe)} \approx 0.5168 < \kapa$. This indicates that the setting $z_2 = 0$ fails to guarantee either $\bi$ or $\ti$ achieves an MPM greater than or equal to the target MPM of $\kapa$, following the approach from~\cite{bib:Jaillet, brubach2020online}.

\section{Proof of the \re} \label{app:re}

\begin{proof}
Consider a neighboring offline node $i' \neq i$, and let the two share one online neighbor $j$. Focus on the impact of $i'$'s performance on the contribution of $j$ to $i$. Without loss of generality (WLOG), assume that $j$ has only two offline neighbors, $i$ and $i'$. Consider a given time $t \in [0,1]$, and assume that $i$ is safe at that time (i.e., $\saf_{i,t}=1$). The matching rate from $j$ for $i'$ remains fixed throughout $[0,t]$, denoted by $a \in \{1/3, 2/3\}$. Let $\alp(t) = \E[\saf_{i,t}]$ and $\psi_{i',t} = \E[\saf_{i',t} | \saf_{i,t}=1]$. We focus on the contribution of the matching rate from the online neighbor $j$ and ignore contributions from all other online neighbors.

\begin{align*}
\alp_{k+1} &= \alp_k \cdot \bp{1 - \frac{1}{T} \bb{1 - \psi_{i',t} + a \psi_{i',t}}} = 
\alp(t) \cdot \bp{1 - \frac{1}{T} \bb{1 - (1-a) \psi_{i',t}}}.
\end{align*}

This implies that
\begin{align*}
\alp'(t) &= -\alp(t) + (1-a) \cdot \alp(t) \cdot \psi_{i'}(t).
\end{align*}

Note that $1-a > 0$ and $\alp(t) > 0$. Thus, maximizing $\alp(t)$ for any $t \in [0,1]$ is equivalent to maximizing $\psi_{i',t} = \E[\saf_{i',t} | \saf_{i,t}=1]$. This means that $i'$ should have the least matching rate from all online neighbors other than $j$, implying that $i'$ should have the largest possible probability of staying safe at time $t$.
\end{proof}

\section{Definition of the \tbf{Folding Procedure}}\label{app:fold}
\begin{definition}
Consider a given randomized graph $G(\X)$ induced by $\X$, where $\X \in \{0,1/3,2/3\}^{|E|}$ denotes the random rounded vector output by $\mathsf{DR}[3]$, and assume every offline node has a unit mass. Given an offline node $\bfi$ with $X_\bfi:=\sum_{e \in E_\bfi} X_e=1$, we define:

\textbf{Type-A edge}: An edge $e=(i,j)$, where $i$ is an offline neighbor of $\bfi$ such that the two share the online neighbor $j$.
 
\textbf{Type-B edge}: An edge $e=(\bi,\bj)$, where $\bi$ is an offline neighbor of $\bfi$, but $\bj$ is not a neighbor of $\bfi$.

\smallskip

Let $e_1=(i,j)$ and $e_2=(\bi,\bj)$ be two distinct edges that are of Type-A and Type-B, respectively, with respect to $\bfi$, satisfying $X_{e_1}=X_{e_2}>0$. A \tbf{Folding Procedure} (\fp) on $e_1$ and $e_2$ is defined as follows:
\begin{enumerate}
    \item Remove edges $e_1$ and $e_2$.
    \item Add a new edge $e=(\bi, j)$ and set $X_e \gets X_{e_1}$ if $e$ does not exist (i.e., $X_e=0$), or update $X_e \gets X_e+X_{e_1}$ if $e$ exists with $X_e>0$ (before \fp). \ensuremath{\hfill\blacksquare}
\end{enumerate}
\end{definition}

\begin{figure}[ht!]
\begin{subfigure}[b]{0.4\textwidth}
\begin{tikzpicture}
 \draw (0,0) node[minimum size=0.2mm,draw,circle, thick] {$i$};
  \draw (0,-1.5) node[minimum size=0.2mm,draw,circle, ultra thick] {$\bfi$};
    \draw (3,-1.5) node[minimum size=0.2mm,draw,circle, thick] {$j$};
    \draw (0,-3) node[minimum size=0.2mm,draw,circle, thick] {$\bi$};
     \draw (3,-3) node[minimum size=0.2mm,draw,circle, thick] {$\tj$};
    \draw (3,-4.5) node[minimum size=0.2mm,draw,circle, thick] {$\bj$};
\draw[-, very thick] (0.4,-1.5)--(2.6,-1.5);
\draw[-, very thick] (0.4,-1.5)--(2.6,-3);
\draw[-, very thick] (0.4,-3)--(2.6,-4.5) node [blue, below, midway, sloped] {$e_2$ (Type B)};
\draw[-,  very thick] (0.4,-3)--(2.6,-3);
\draw[-,  very thick] (0.4,0)--(2.6,-1.5) node [blue, above, midway, sloped] {$e_1$ (Type A)};
\end{tikzpicture}
\caption{\tbf{Before} folding procedure.}
 \label{fig:09-08-a}
\end{subfigure}
\begin{subfigure}[b]{0.4\textwidth}
\centering
\begin{tikzpicture}
 \draw (0,0) node[minimum size=0.2mm,draw,circle, thick] {$i$};
  \draw (0,-1.5) node[minimum size=0.2mm,draw,circle, ultra thick] {$\bfi$};
    \draw (3,-1.5) node[minimum size=0.2mm,draw,circle, thick] {$j$};
    \draw (0,-3) node[minimum size=0.2mm,draw,circle, thick] {$\bi$};
     \draw (3,-3) node[minimum size=0.2mm,draw,circle, thick] {$\tj$};
    \draw (3,-4.5) node[minimum size=0.2mm,draw,circle, thick] {$\bj$};
\draw[-, very thick] (0.4,-1.5)--(2.6,-1.5);
\draw[-, very thick] (0.4,-1.5)--(2.6,-3);
\draw[loosely dashed, very thick] (0.4,-3)--(2.6,-4.5) node [blue, below, midway, sloped] {$e_2$ (Type B)};
\draw[densely dashdotted,  ultra thick] (0.4,-3)--(2.6,-1.5);
\draw[-,  very thick] (0.4,-3)--(2.6,-3);

\draw[loosely dashed,  very thick] (0.4,0)--(2.6,-1.5) node [blue, above, midway, sloped] {$e_1$ (Type A)};
\end{tikzpicture}
\caption{\tbf{After} folding procedure.}  \label{fig:09-08-b}
\end{subfigure}
\caption{A generic illustration of the \tbf{Folding Procedure} (\fp) applied to $e_1=(i,j)$ of Type-A and $e_2=(\bi, \bj)$ of Type-B with respect to the target node $\bfi$ such that $i \neq \bi$ and $X_{e_1}=X_{e_2}$. Black edges here can be either big (of value $2/3$) or small (of value $1/3$), and irrelevant edges may be omitted for some nodes. The \tbf{Folding Procedure} consists of two steps: (1) Remove $e_1$ and $e_2$, both of which are shown in a loosely dashed style in Figure~\ref{fig:09-08-b}; and (2) Add a new edge $e=(\bi,j)$ and set $X_e=X_{e_1}$ if $X_e=0$, or update $X_e \gets X_e+X_{e_1}$ if $X_e>0$, where $e$ is shown in a densely dash-dotted style in Figure~\ref{fig:09-08-b}.}
\label{fig:09-08-c}
\end{figure}
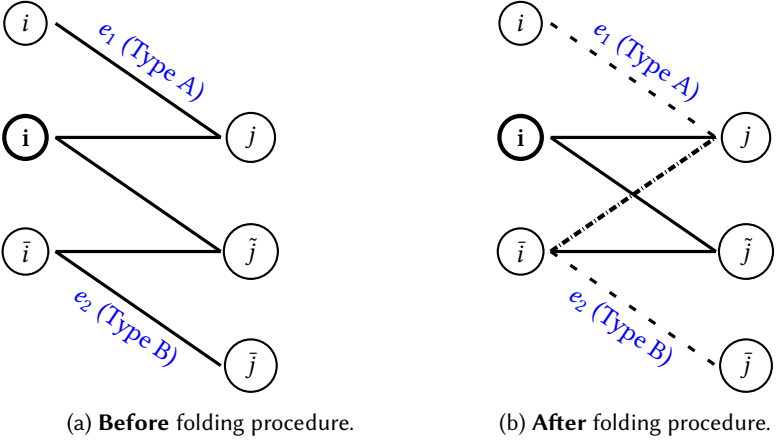

\xhdr{Remarks on the Folding Procedure (\fp)}: 
\begin{enumerate}
    \item \tbf{Pre-existing Edge}: When $e = (\bi, j)$ already exists in Figure~\ref{fig:09-08-a} with $X_e > 0$, we know for sure that $X_e = X_{e_1} = X_{e_2} = 1/3$ before \fp. This is because node $j$ has three distinct edges in $\X$ (since $\bfi$, $i$, and $\bi$ are all distinct). After \fp, $X_e$ becomes $2/3$.

    \item \tbf{Properties of \fp}: The Folding Procedure (\fp) never changes the structure of the target node $\bfi$. While it may affect the structures of $\bfi$'s offline neighbors, $\bfi$ remains the same type (e.g., 3S or 1B1S) after \fp. Meanwhile, the probability of the target node $\bfi$ staying safe will never decrease after \fp (alternatively, the probability of $\bfi$ getting matched will never increase after \fp).

    \item \tbf{Purpose of the Folding Procedure}: The Folding Procedure is solely for the WS competitive analysis of the target node $\bfi$ when all offline nodes have a unit mass after rounding. Its goal is to pinpoint the WS structure of the target node that maximizes its probability of staying available (or minimizes the probability of being matched). Unlike the \tbf{Modification Procedures} proposed for the case when a mass-one node has offline neighbors of mass less than one, as shown in Section~\ref{sec:mod},  \fp is never implemented as part of the algorithm. \textit{As a result, there is no concern that it might increase the likelihood of any offline node becoming a 1B1S type after rounding}.

    \item \tbf{Special Case}: The case shown in Figure~\ref{fig:11/15/23/a} is a special instance of \fp, as illustrated in Figure~\ref{fig:09-08-c}. Using a similar argument, we can demonstrate that the result in Lemma~\ref{lem:11/15/a} continues to hold for this general case.
\end{enumerate}

\subsection{Proof of Lemma~\ref{lem:ws_one}}

\begin{proof}
When the target offline node $\bfi$ is of type 1B1S (one big and one small edge), applying the \tbf{Folding Procedure} and \re results in a single possible WS structure, as depicted in Figure~\ref{fig:three/a}. Lemma~\ref{lem:11/15/a} and the accompanying illustrations in Figure~\ref{fig:11/15/23/a} provide a complete proof for this case.

Now, let us focus on a target offline node $\bfi$ of type 3S (three small edges). This node has three online neighbors: $\tj$, $\bj$, and $j$, as shown in Figure~\ref{fig:09-09-c}. Due to the \re, these online neighbors are expected to exhibit symmetric structures. Define the following pairs of edges:  
\[
\cP_1 = (\te_1, \te_2), \quad \cP_2 = (\bae_1, \bae_2), \quad \cP_3 = (\te_3, \te_4), \quad \text{and} \quad \cP_4 = (\bae_3, \bae_4).
\]
Note that the target node $\bfi$ must have at least two offline neighbors. By symmetry, we can assume without loss of generality (WLOG) that when applying the \tbf{Folding Procedure} (\fp), we aim to prune edges in $\cP_\ell$ for $\ell \in \{1,2,3,4\}$ and to add or update values on edges between $a \in \{\bi, \ti\}$ and $b \in \{\bj, \tj\}$. We analyze this scenario by considering the following cases.

---

\medskip

\noindent \tbf{Case 1:} Edges within each pair $\cP_\ell$ with $\ell \in \{1,2,3,4\}$ are identical.  
This implies that both $\tj$ and $\bj$ each have a big edge (of mass 2/3) in addition to the one connected to $\bfi$. Similarly, both $\ti$ and $\bi$ each have a big edge in addition to the one connected to $j$. Applying \fp to the big edge $e_c$ (the merged version of $\te_1$ and $\te_2$) and another big edge $e_d$ (the merged version of $\te_3$ and $\te_4$), we add a big edge between $\ti$ and $\tj$. Repeating this process, we add another big edge between $\bi$ and $\bj$. Consequently, $\bfi$ ends up with the structure shown in Figure~\ref{fig:three/c}.

---

\medskip

\noindent \tbf{Case 2:} Edges within only one pair of $\cP_1$ or $\cP_2$ are identical, and by symmetry, edges within only one pair of $\cP_3$ or $\cP_4$ are identical.  
WLOG, assume that edges in $\cP_1$ and $\cP_3$ are identical, while edges in $\cP_2$ and $\cP_4$ are distinct. Applying \fp to $e_c$ (the merged edges in $\cP_1$) and $e_d$ (the merged edges in $\cP_3$), we add a new big edge between $\ti$ and $\tj$. Now, consider the following two subcases.

---

\noindent \tbf{Case 2a:} One edge from $\cP_2$, say $\bae_1$, is identical to another edge from $\cP_4$, say $\bae_3$, while $\bae_2$ differs from $\bae_4$.  
Applying \fp to $\bae_2$ and $\bae_4$ results in a single big edge between $\bi$ and $\bj$.

---

\noindent \tbf{Case 2b:} No edge in $\cP_2$ is identical to any edge in $\cP_4$.  
Applying \fp to $\bae_1$ and $\bae_3$, followed by another round of \fp to the remaining edges $\bae_2$ and $\bae_4$, results in a single big edge between $\bi$ and $\bj$.

---

Wrapping up the subcases above, we conclude that for \tbf{Case 2}, the target node $\bfi$ ends up with the structure shown in Figure~\ref{fig:three/c}.

---

\medskip

\noindent \tbf{Case 3:} No edges in the four pairs $\cP_\ell$ with $\ell \in \{1,2,3,4\}$ are identical.  
We consider the following subcases.

---

\noindent \tbf{Case 3a:} The pairs $\cP_1$ and $\cP_3$ share one edge (but cannot share two edges, as this would imply identical edges in both pairs).  
WLOG, let $\te_1 = \te_3$ and $\te_2 \neq \te_4$. Similarly, assume that $\cP_2$ and $\cP_4$ share only one edge, so $\bae_1 = \bae_3$ and $\bae_2 \neq \bae_4$. Applying \fp multiple times yields two possible outcomes:
1. Applying \fp to $\te_2$ and $\te_4$, followed by another round of \fp to the remaining edges $\bae_2$ and $\bae_4$, results in two big edges: $(\ti, \tj)$ and $(\bi, \bj)$. In this case, $\bfi$ ends up with the structure shown in Figure~\ref{fig:three/c}.  
2. Alternatively, applying \fp to $\te_2$ and $\bae_4$, followed by another round of \fp to $\bae_2$ and $\te_4$, results in $\bfi$ having the structure shown in Figure~\ref{fig:three/b}.

---

\noindent \tbf{Case 3b:} No edge in $\cP_1 \cup \cP_2$ is identical to any edge in $\cP_3 \cup \cP_4$.  
By applying similar analyses as shown above, we conclude that $\bfi$ can end up with a structure as shown in either Figure~\ref{fig:three/b} or Figure~\ref{fig:three/c}.

---

\noindent This completes the proof of Lemma~\ref{lem:ws_one}.
\end{proof}

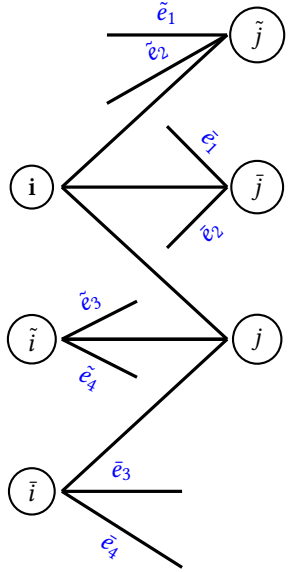
\begin{figure}[th!]
\begin{tikzpicture}
 \draw (0,0) node[minimum size=0.2mm,draw,circle, thick] {$\bfi$};
    \draw (3,0) node[minimum size=0.2mm,draw,circle, thick] {$\bj$};
     \draw[-, very thick] (1.8,0.8)--(2.6,0) node [blue, above, midway, sloped] {$\bar{e}_1$};
      \draw[-, very thick] (1.8,-0.8)--(2.6,0) node [blue, below, midway, sloped] {$\bar{e}_2$};
    
        \draw (3,-2) node[minimum size=0.2mm,draw,circle, thick] {$j$};
        
 \draw (3, 2) node[minimum size=0.2mm,draw,circle, thick] {$\tj$};
 
 \draw[-, very thick] (0.4,0)--(2.6,2);
 \draw[-, very thick] (1,2)--(2.6,2) node [blue, above, midway, sloped] {$\te_1$};
 \draw[-, very thick] (1,1.1)--(2.6,2) node [blue, above, midway, sloped] {$\te_2$};

 \draw[-, very thick] (0.4,-2)--(1.4,-1.5) node [blue, above, midway, sloped] {$\te_3$};
 \draw[-, very thick] (0.4,-2)--(1.4,-2.5) node [blue, below, midway, sloped] {$\te_4$};
 
  \draw[-, very thick] (0.4,-4)--(2, -4) node [blue, above, midway, sloped] {$\bae_3$};
 \draw[-, very thick] (0.4,-4)--(2,-5) node [blue, below, midway, sloped] {$\bae_4$};

 \draw[-, very thick] (0.4,-2)--(2.6,-2);
  \draw[-, very thick] (0.4,-2)--(2.6,-2);

\draw (0,-2) node[minimum size=0.2mm,draw,circle, thick] {$\ti$};
\draw (0,-4) node[minimum size=0.2mm,draw,circle, thick] {$\bi$};

\draw[-, very thick] (0.4,-4)--(2.6,-2);

\draw[-, very thick] (0.4,0)--(2.6,0);
\draw[-, very thick] (0.4,0)--(2.6,-2);
\end{tikzpicture}
  \caption{Pinpointing all possible WS structures of an offline node of type 3S when all offline nodes have a unit mass after rounding.}
  \label{fig:09-09-c}
\end{figure}

\section{Another Markov Chain-Based Approach for the Example in Figure~\ref{fig:three/c}}\label{app:three/c}

Recall that the discrete version of our arrival setting is as follows: There are $T$ rounds, and during each round $k \in [T]$, each online node $j \in J$ arrives uniformly with probability $1/T$. Consider a Markov Chain (MC) with six states defined as follows:
\[
s_1 = (2,1), \quad s_2 = (1,1), \quad s_3 = (2,0), \quad s_4 = (0,1), \quad s_5 = (1,0), \quad s_6 = (0,0),
\]
where in each state, the first value represents the number of safe nodes between $i_0$ and $i_2$, and the second value represents the number of safe nodes corresponding to $i_1$. The process starts in state $s_1$ with probability one, and the one-step transition matrix is given below:
\[
\mathbf{H} = (H_{pq}) =
\begin{pmatrix}
1 - 3/T & 2/T & 1/T & 0 & 0 & 0 \\
0 & 1 - 3/T & 0 & (7/6)/T & (11/6)/T & 0 \\
0 & 0 & 1 - 3/T & 0 & 3/T & 0 \\
0 & 0 & 0 & 1 - 3/T & 0 & 3/T \\
0 & 0 & 0 & 0 & 1 - 2/T & 2/T \\
0 & 0 & 0 & 0 & 0 & 1
\end{pmatrix},
\]
where $H_{pq} = \Pr[X_{k+1} = s_p \mid X_k = s_q]$ denotes the one-step transition probability from state $s_q$ to state $s_p$, for all $p, q \in \{1, 2, 3, 4, 5, 6\}$.

\smallskip

By computing the limit of $\mathbf{H}^T$ (the $T$th power of $\mathbf{H}$) as $T \to \infty$, we obtain the final stationary distribution $\bpi = (\pi_p)$ over the six states as follows:
\[
\bpi = \left(
\frac{1}{\sfe^3}, \frac{2}{\sfe^3}, \frac{1}{\sfe^3}, \frac{7}{6 \sfe^3}, \frac{20 (\sfe - 2)}{3 \sfe^3}, \frac{49 - 40 \sfe + 6 \sfe^3}{6 \sfe^3}
\right).
\]

Let $N_\ell$ denote the expected number of matches for $i_\ell$ at the end of the process, where $\ell \in \{0, 1, 2\}$. We compute:
\begin{align*}
\E[N_0] &= \E[N_2] = \frac{\pi_2 + \pi_5}{2} + \pi_4 + \pi_6 = \frac{11 - 10 \sfe + 3 \sfe^3}{3 \sfe^3} \approx 0.7314; \\
\E[N_1] &= \pi_3 + \pi_5 + \pi_6 = 1 - \frac{25}{6 \sfe^3} \approx 0.7925.
\end{align*}

\section{Alterations on Sampling Distributions for Mass-One Offline Nodes}\label{sec:7/7/a}
\subsection{Analysis of a Mass-One Node with Combined Structures (Figure~\ref{fig:6/29/a}+\ref{fig:6/29/b})}
\begin{figure}[ht!]
\begin{tikzpicture}
 \draw (0,0) node[minimum size=0.2mm,draw,circle, ultra thick] {$\mathbf{i}$};
   \draw (-0.4,1) node[left] {\bluee{$1/3$}};
  \draw (-0.4,0) node[left] {\bluee{$1$}};
    \draw (-0.4,-1) node[left] {\bluee{$2/3$}};
  \draw (3,0) node[minimum size=0.2mm,draw,circle, thick] {$j_s$};
  \draw (3,1) node[minimum size=0.2mm,draw,circle, thick] {$j_b$};
 \draw (0,-1) node[minimum size=0.2mm,draw,circle, thick] {$\ti$};

    \draw (0,1) node[minimum size=0.2mm,draw,circle, thick] {${\hi}$};

\draw[-,  very thick] (0.4,1)--(2.6,1);
\draw[-,  very thick] (0.4,0)--(2.6,0) node [blue, above, near end, sloped] {$1-z_2$};
\draw[-,  red, ultra thick] (0.4, -1)--(2.6,0) node [blue, above, midway, sloped] {$z_2$};

  \draw (1.5,1) node[above] {\bluee{$z_1$}};
\draw[-,  red, ultra thick] (0.4, 0)--(2.6,1) node [blue, above, near start, sloped] {$1-z_1$};   
\end{tikzpicture}
\caption{The offline node $\bfi$ (of mass one) has two online neighbors: one with the structure shown in Figure~\ref{fig:6/29/a} and the other in Figure~\ref{fig:6/29/b}. The values on the edges represent the updates proposed to achieve \tbf{Goals 1 and 2} outlined in Section~\ref{sec:tworep}.}
    \label{fig:7/4/a}
    \end{figure}
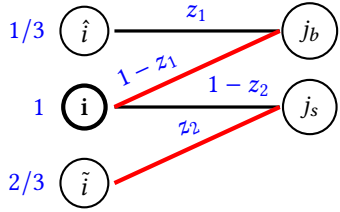
Note that the original case corresponds to when $z_1=1/3$ and $z_2=2/3$. The lemma below suggests that the simple setting of $z_1=0$ and $z_2=1$ suffices to achieve both \tbf{Goal 1} and \tbf{Goal 2}. 

\begin{lemma}
The configuration of $z_1=0$ and $z_2 \in [0.5663, 1]$ in Figure~\ref{fig:7/4/a} guarantees achieving both \tbf{Goal 1} and \tbf{Goal 2}, as stated in Section~\ref{sec:tworep}.
\end{lemma}

\begin{proof}
Consider the setting of $z_1=0$ and $z_2=z$. Let us investigate the conditions we should impose on $z$ such that \tbf{Goal 1} and \tbf{Goal 2} are satisfied. First, we analyze $\ti$. Let $\alp(t)$ be the probability that $\ti$ is safe at $t \in [0,1]$, and let $\beta(t)$ be the probability that both $\bfi$ and $\ti$ are safe at $t$. Thus, $\beta(t)=\sfe^{-2t}$.

\begin{align*}
&\alp_{k+1} = \alp_k \cdot \bP{1-\frac{1}{T}\bp{\frac{\beta}{\alp} \cdot z + \bp{1-\frac{\beta}{\alp}} \cdot 1}}, \quad \alp_1 = 1. \\
\Rightarrow~~ &\alp' + \alp = \beta (1 - z) = \sfe^{-2t} (1 - z), \quad \alp(0) = 1. \\
\Rightarrow~~ &\alp(t) = \sfe^{-2t} \cdot \bp{2\sfe^t - 1 - z \sbp{\sfe^t - 1}}, \quad t \in [0,1].
\end{align*}

Therefore, the condition that $\ti$ achieves an MPM of at least $\kap_m^* \approx 0.7969$ is equivalent to
\[
\sbp{1 - \alp(1)}/(2/3) = \bP{1 - \sfe^{-2} \bp{2\sfe - 1 - z (\sfe - 1)}}/(2/3) \ge \kap_m^* \Leftrightarrow z \ge 0.5663.
\]

Next, we analyze the target node $\bfi$. Let $\alp(t)$ be the probability that $\bfi$ is safe at $t$, and let $\beta(t)$ be the probability that both $\bfi$ and $\ti$ are safe at $t$.

\begin{align*}
&\alp_{k+1} = \alp_k \cdot \bP{1 - \frac{1}{T} \bp{1 + \frac{\beta}{\alp} \cdot (1 - z) + \bp{1 - \frac{\beta}{\alp}} \cdot 1}}, \quad \alp_1 = 1. \\
\Rightarrow~~ &\alp' + 2\alp = \beta \cdot z = \sfe^{-2t} \cdot z, \quad \alp(0) = 1. \\
\Rightarrow~~ &\alp(t) = \sfe^{-2t} \cdot \sbp{1 + t \cdot z}, \quad t \in [0,1].
\end{align*}

To ensure that $\bfi$ achieves an MPM of at least $\kap_B^* = 1 - 2\sfe^{-2}$, we require
\[
1 - \alp(1) = 1 - \sfe^{-2}(1 + z) \ge 1 - 2\sfe^{-2} \Leftrightarrow z \le 1.
\]

Finally, we analyze $\hi$. Let $\alp(t)$ be the probability that $\hi$ is safe at $t$, and let $\beta(t)$ be the probability that both $\hi$ and $\bfi$ are safe at $t$. Note that $\beta(t) = \sfe^{-2t} \cdot \sbp{1 + t \cdot z}$. We have

\begin{align*}
&\alp' + \alp = \beta, \quad \alp(0) = 1. \\
\Rightarrow~~ &\alp(t) = \sfe^{-2t} \bp{2\sfe^t - 1 + z \sbp{\sfe^t - t - 1}}.
\end{align*}

Therefore,
\[
\sbp{1 - \alp(1)}/(1/3) = \bp{1 - \sfe^{-2} \sbp{2\sfe - 1 + z (\sfe - 2)}}/(1/3) \ge \kap_s^* \Leftrightarrow z \le 1.
\]

Thus, any setting of $z_1 = 0$ and $z_2 = z \in [0.5663, 1]$ suffices to achieve \tbf{Goal 1} and \tbf{Goal 2}.
\end{proof}

\subsection{Analysis of a Mass-One Node with Combined Structures ($3 \times$~Figure~\ref{fig:6/29/b})}

\begin{figure}[ht!]
\begin{tikzpicture}
 \draw (0,0) node[minimum size=0.2mm,draw,circle, ultra thick] {$\mathbf{i}$};
   \draw (-0.4,1) node[left] {\bluee{$2/3$}};
  \draw (-0.4,0) node[left] {\bluee{$1$}};
    \draw (-0.4,-1) node[left] {\bluee{$2/3$}};
        \draw (-0.4,-2) node[left] {\bluee{$2/3$}};
  \draw (3,0) node[minimum size=0.2mm,draw,circle, thick] {$j_2$};
  \draw (3,1) node[minimum size=0.2mm,draw,circle, thick] {$j_1$};
    \draw (3,-2) node[minimum size=0.2mm,draw,circle, thick] {$j_3$};
 \draw (0,-1) node[minimum size=0.2mm,draw,circle, thick] {$\ti$};
 \draw (0,-2) node[minimum size=0.2mm,draw,circle, thick] {$\bi$};
    \draw (0,1) node[minimum size=0.2mm,draw,circle, thick] {${\hi}$};

\draw[-,  red, ultra thick] (0.4,1)--(2.6,1);
\draw[-,  very thick] (0.4,0)--(2.6,0) node [blue, above, near end, sloped] {$1-z$};
\draw[-,  very thick] (0.4,0)--(2.6,-2) node [blue, above, near end, sloped] {$1-z$};
\draw[-,  red, ultra thick] (0.4, -1)--(2.6,0) node [blue, below, near end, sloped] {$z$};
\draw[-,  red, ultra thick] (0.4, -2)--(2.6,-2) node [blue, above, near end, sloped] {$z$};

  \draw (1.5,1) node[above] {\bluee{$z$}};
\draw[-, very thick] (0.4, 0)--(2.6,1) node [blue, above, near start, sloped] {$1-z$};   
\end{tikzpicture}
\caption{The offline node $\bfi$ (of mass one) has three online neighbors, each having the structure shown in Figure~\ref{fig:6/29/b}. The values on the edges represent the updates proposed to achieve \tbf{Goals 1 and 2}  in Section~\ref{sec:tworep}.}
    \label{fig:7/1/1}
    \end{figure}
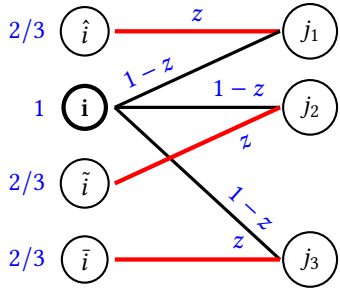

\begin{lemma}
The configuration of $z \in [z^*, \sqrt[3]{9/2} - 1, 1]$ in Figure~\ref{fig:7/1/1} guarantees achieving both \tbf{Goal 1} and \tbf{Goal 2}, as stated in Section~\ref{sec:tworep}, where $z^* \approx 0.5607$ is the unique solution to the following equation:
\[
\frac{2 + 4 \sfe^3 + 4 z - z^2 - 5 z^3 + \sfe^2 (-6 + z^2 + z^3)}{4 \sfe^3} = \frac{2}{3} \cdot \kap^*_m.
\]
\end{lemma}

\begin{proof}
Consider the following continuous-time Markov chain, which consists of eight states:
\[
s_1 = (1,3), \quad s_2 = (0,3), \quad s_3 = (1,2), \quad s_4 = (1,1), \quad s_5 = (0,2), \quad s_6 = (1,0), \quad s_7 = (0,1), \quad s_8 = (0,0),
\]
where in each state, the first number being equal to $1$ indicates that $\bfi$ is safe, and the second number represents the number of safe nodes among the three offline neighbors of $\bfi$. The transition-rate matrix is stated in Table~\ref{table:7/1/1}.

\begin{table}[ht!]
\caption{The transition-rate matrix of the continuous-time Markov chain for the instance in Figure~\ref{fig:7/1/1}.}
\label{table:7/1/1}
\centering
\begin{tabular}{c " c ccccccc} \thickhline
         & $(1,3)$ & $(0,3)$ & $(1,2)$ & $(1,1)$ & $(0,2)$ & $(1,0)$ & $(0,1)$ & $(0,0)$   \\ \thickhline
$(1,3)$ & $-3$ & $3(1-z)$ & $3z$ & 0 & 0 & 0 & 0 & 0 \\
$(0,3)$ & 0 & $-3$ & 0 & 0 & 3 & 0 & 0 & 0 \\
$(1,2)$ & 0 & 0 & $-3$ & $2z$ & $3 - 2z$ & 0 & 0 & 0 \\
$(1,1)$ & 0 & 0 & 0 & $-3$ & 0 & $z$ & $3 - z$ & 0 \\
$(0,2)$ & 0 & 0 & 0 & 0 & $-2$ & 0 & 2 & 0 \\
$(1,0)$ & 0 & 0 & 0 & 0 & 0 & $-3$ & 0 & 3 \\
$(0,1)$ & 0 & 0 & 0 & 0 & 0 & 0 & $-1$ & 1 \\
$(0,0)$ & 0 & 0 & 0 & 0 & 0 & 0 & 0 & 0 \\
\end{tabular}
\end{table}

The Kolmogorov forward equations, along with the initial conditions for the process, are stated below:
\begingroup
\allowdisplaybreaks
\begin{align*}
q_1'(t) &= -3 q_1(t), \\
q_2'(t) &= 3 (1 - z) q_1(t) - 3 q_2(t), \\
q_3'(t) &= 3 z q_1(t) - 3 q_3(t), \\
q_4'(t) &= 2 z q_3(t) - 3 q_4(t), \\
q_5'(t) &= 3 q_2(t) + (3 - 2z) q_3(t) - 2 q_5(t), \\
q_6'(t) &= z q_4(t) - 3 q_6(t), \\
q_7'(t) &= (3 - z) q_4(t) + 2 q_5(t) - q_7(t), \\
q_8'(t) &= 3 q_6(t) + q_7(t), \\
q_1(0) &= 1, \quad q_\ell(0) = 0, \quad \forall 2 \le \ell \le 8.
\end{align*}
\endgroup

Let $\eta_1(z)$ and $\eta_{2/3}(z)$ denote the MPM achieved for the node $\bfi$ and each of its offline neighbors, respectively. Thus, we have
\begingroup
\allowdisplaybreaks
\begin{align*}
\eta_1(z) &= q_2(1) + q_5(1) + q_7(1) + q_8(1) = 1 - \sfe^{-3} \cdot (1 + z)^3 \ge \kap_S^*,
\end{align*}
\endgroup
where we can solve that $z \le \sqrt[3]{9/2} - 1 \approx 0.6510$. Similarly, we have
\begin{align*}
\eta_{2/3}(z) &= \left[\frac{q_3(1) + q_5(1)}{3} + \frac{q_4(1) + q_7(1)}{3} \cdot 2 + q_6(1) + q_8(1)\right] \cdot \frac{1}{2/3} \ge \kap_m^*,
\end{align*}
where we can solve that $z \ge 0.5607$.
\end{proof}

\subsection{Analysis of a Mass-One Node with Combined Structures ($3 \times$~Figure~\ref{fig:6/29/c})}
\begin{figure}[ht!]
\begin{tikzpicture}
 \draw (0,0) node[minimum size=0.2mm,draw,circle, ultra thick] {$\mathbf{i}$};
  \draw (-0.4,0) node[left] {\bluee{$1$}};
    \draw (-0.4,-1) node[left] {\bluee{$1/3$}};
    \draw (-0.4,-2) node[left] {\bluee{$1/3$}};
    
        \draw[-,  very thick] (1.5,1)--(2.6,1) node [blue, above, midway, sloped] {$z$};
     \draw[-,  very thick] (2,1.6)--(2.6,1)  node [blue, above, midway, sloped] {$z$};
     
    \draw[-,  very thick] (1.5,3)--(2.6,3)node [blue, above, midway, sloped] {$z$};
     \draw[-,  very thick] (1.5,2.5)--(2.6,3)node [blue, above, near start, sloped] {$z$};
     
         \draw (3,3) node[minimum size=0.2mm,draw,circle, thick] {$j''$};
     \draw (3,1) node[minimum size=0.2mm,draw,circle, thick] {$j'$};
 \draw (3,0) node[minimum size=0.2mm,draw,circle, thick] {$j$};
 
 \draw (0,-1) node[minimum size=0.2mm,draw,circle, thick] {$\bi$};
  \draw (0,-2) node[minimum size=0.2mm,draw,circle, thick] {$\ti$};

\draw[-,  very thick] (0.4,0)--(2.6,3) node [blue, above, midway, sloped] {$1-2z$};
\draw[-,  very thick] (0.4,0)--(2.6,1) node [blue, above, midway, sloped] {$1-2z$};
\draw[-,  very thick] (0.4,0)--(2.6,0) node [blue, above, near end, sloped] {$1-2z$};
\draw[-,  very thick] (0.4, -1)--(2.6,0) node [blue, above, midway, sloped] {$z$};
\draw[-,  very thick] (0.4, -2)--(2.6,0) node [blue, below, midway, sloped] {$z$};  
\end{tikzpicture}
\caption{The offline node $\bfi$ (of mass one) has three online neighbors, each having the structure shown in Figure~\ref{fig:6/29/c}. The values on the edges represent the updates proposed to achieve \tbf{Goals 1 and 2}  in Section~\ref{sec:tworep}.}
    \label{fig:7/1/3}
    \end{figure}
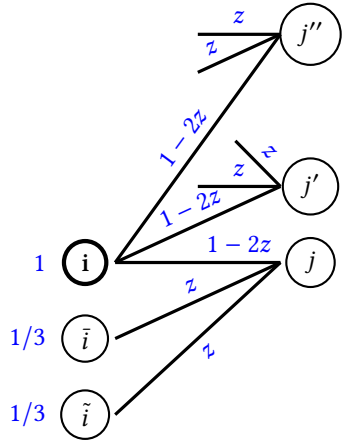

\begin{lemma}
The configuration of $z = 0$ in Figure~\ref{fig:7/1/3} guarantees achieving both \tbf{Goal 1} and \tbf{Goal 2}, as stated in Section~\ref{sec:tworep}.
\end{lemma}

\begin{proof}
First, for the mass-one node $\bfi$, it achieves an MPM equal to $1 - \sfe^{-3} > \kap^*_S$. Now, we show that each of its offline neighbors achieves an MPM at least equal to that in Figure~\ref{fig:6/30/1} (with $z_1 = z_2 = 0$). Observe that the probability of $\bfi$ being safe at time $t$ is $\sfe^{-3t}$, which is strictly smaller than $\sfe^{-2t}$, the probability of $\bfi$ being safe at $t$ in Figure~\ref{fig:6/30/1} with $z_1 = z_2 = 0$.

\smallskip

Focus on the module consisting of $\bfi$, $\bi$, and $\ti$. We can consider the two offline nodes $\bi$ and $\ti$ together as a bundle. In other words, the bundle $(\bi, \ti)$ can be treated as a single entity for the purposes of the matching process. Note that when $\bfi$ is safe, the matching rate for either $\bi$ or $\ti$ is zero. Conversely, when $\bfi$ is not safe, the matching rate becomes one if at least one of $\bi$ or $\ti$ is safe. The same argument applies to the module of the triple $(\bfi, \bi, \ti)$ in Figure~\ref{fig:6/30/1} with $z_1 = z_2 = 0$.

\smallskip

Thus, we claim that the matching process for the bundle $(\bi, \ti)$ starts earlier in expectation in Figure~\ref{fig:7/1/3} (with $z = 0$) compared to Figure~\ref{fig:6/30/1} (with $z_1 = z_2 = 0$). Once the matching process begins, the remaining steps are identical. As a result, we establish our claim.
\end{proof}

\section{Analysis of the WS Structures for Offline Nodes of Mass 1/3 and 2/3}\label{app:ws/12/3}

For each offline node of mass $a \in \{1/3, 2/3\}$, let $\psi(a)$ and $\phi(a) = \psi(a)/a$ denote the corresponding matching probability and Matching Probability per Mass (MPM), respectively. Meanwhile, let $\tau(t)$ denote the probability that $\ti$ (or $\hi$) is safe at time $t \in [0,1]$ given that $\bfi$ is safe at that time, and let $\eta(t)$ denote the matching rate of $\bfi$ from $j$ given that $\bfi$ is safe.

\xhdr{Case 1}: The target node has a mass of $1/3$ and a structure as shown in Figure~\ref{fig:wsg/3}.
\begingroup
\allowdisplaybreaks
\begin{align*}
&\psi(1/3) = 1 - \frac{5}{3\sfe}, \quad &&\psi(2/3) = 1 - \frac{4}{3\sfe},\\
&\phi(1/3) = 3 - \frac{5}{\sfe} \approx 1.1606, \quad &&\phi(2/3) = \frac{3}{2} - \frac{2}{\sfe} \approx 0.7642,\\
&\tau(t) = \frac{1}{1 + \frac{2}{3} t}, \quad &&\eta(t) = \frac{\tau(t)}{3} + 1 - \tau(t) = 1 - \frac{2}{3} \tau(t).
\end{align*}
\endgroup

\xhdr{Case 2}: The target node $\bfi$ has a mass of $1/3$ and a structure as shown in Figure~\ref{fig:wsg/4}. For any time $t \in [0,1]$ and any offline node $\ell \in \{\bfi, \ti, \hi\}$, let $\saf_{\ell,t} = 1$ indicate that $\ell$ is safe at $t$. In this case, we introduce two versions of $\tau(t)$, namely, $\tau_1(t)$ and $\tau_2(t)$, where:
\begin{itemize}
    \item $\tau_1(t)$ denotes the probability that, given $\bfi$ is safe at $t$, exactly one of its offline neighbors (either $\hi$ or $\ti$) is safe at $t$.
    \item $\tau_2(t)$ denotes the probability that, given $\bfi$ is safe at $t$, both of its offline neighbors ($\hi$ and $\ti$) are safe at $t$.
\end{itemize}

\begingroup
\allowdisplaybreaks
\begin{align*}
\tau_2(t) &= \E[\saf_{\ti,t} \cdot \saf_{\hi,t} \mid \saf_{\bfi,t} = 1] = \frac{\E[\saf_{\ti,t} \cdot \saf_{\hi,t} \cdot \saf_{\bfi,t}]}{\E[\saf_{\bfi,t}]}\\
&= \frac{\sfe^{-t}}{\sfe^{-t} \left(1 + \frac{2}{3} t + \frac{t^2}{2} \cdot \frac{2}{3} \cdot \frac{1}{2}\right)} = \frac{1}{1 + \frac{2}{3} t + \frac{1}{6} t^2},\\
\frac{\tau_1(t)}{2} &= \E[\saf_{\ti,t} \cdot (1 - \saf_{\hi,t}) \mid \saf_{\bfi,t} = 1] = \frac{\E[\saf_{\bfi,t} \cdot \saf_{\ti,t} \cdot (1 - \saf_{\hi,t})]}{\E[\saf_{\bfi,t}]}\\
&= \frac{\sfe^{-t}/3}{\sfe^{-t} \left(1 + \frac{2}{3} t + \frac{t^2}{2} \cdot \frac{2}{3} \cdot \frac{1}{2}\right)} = \frac{1}{3 + 2t + \frac{1}{2} t^2}.
\end{align*}
\endgroup

Thus, we have:
\begingroup
\allowdisplaybreaks
\begin{align*}
\psi(1/3) &= 1 - \frac{11}{6\sfe}, \quad \phi(1/3) = 3 - \frac{11}{2\sfe} \approx 0.9766,\\
\eta(t) &= \tau_1(t) \cdot \frac{1}{2} + \tau_2(t) \cdot \frac{1}{3} + 1 - \tau_1(t) - \tau_2(t) = 1 - \frac{\tau_1(t)}{2} - \frac{2}{3} \tau_2(t).
\end{align*}
\endgroup

\xhdr{Case 3}: The target node $\bfi$ is of mass $2/3$. When $\bfi$ has one single big edge as shown in Figure~\ref{fig:wsg/1}, it performs identically to $\ti$ in Figure~\ref{fig:wsg/3}. Now, focus on the case where $\bfi$ has two small edges, each having a structure in Figure~\ref{fig:wsg/3} or Figure~\ref{fig:wsg/4}. By Lemma~\ref{lem:2/3}, we see that $\bfi$'s performance when each online neighbor has a structure in Figure~\ref{fig:wsg/4} is strictly worse than that in Figure~\ref{fig:wsg/3}. Below, we apply a Markov Chain (MC)-based approach to compute the exact performance for the former.

\smallskip

Consider the following discrete Markov Chain (MC). Let $(p,q)$ be the state, where $p = 1$ represents that $\bfi$ is safe, and $p = 0$ otherwise, while $q \in \{4,3,2,1,0\}$ represents the number of $\bfi$'s safe offline neighbors in total. To simplify our computation, we combine a few states as follows:
\begin{itemize}
    \item Let $(0,*) = \{(p=0, q) \mid q \in \{4,3,2,1,0\}\}$ denote the state that $\bfi$ is matched, and treat it as a terminal state.
    \item Let $(1,2N)$ denote the case when there are two safe offline neighbors of $\bfi$ that share the same online neighbor, while $(1,2S)$ denotes the case when the two safe offline neighbors of $\bfi$ are connected to different online neighbors.
\end{itemize}

In this way, we create an MC that starts at state $(1,4)$ and ends at $(0,*)$, with other transient states consisting of $(1,3)$, $(1,2N)$, $(1,2S)$, $(1,1)$, and $(1,0)$. The corresponding one-step transition matrix is outlined in Table~\ref{table:one-step-a}. The stationary distribution over all states when $T \to \infty$ is as follows:
\[
\bpi = \left(\frac{1}{\sfe^2}, \frac{4}{3\sfe^2}, \frac{1}{3\sfe^2}, \frac{4}{9\sfe^2}, \frac{2}{9\sfe^2}, \frac{1}{36\sfe^2}, 1 - \frac{121}{36\sfe^2}\right).
\]

\begin{table}[ht!]
\caption{One-step transition matrix.}\label{table:one-step-a}
\centering
\begin{tabular}{c " c cccccc} \thickhline 
         & $(1,4)$ & $(1,3)$ & $(1,2N)$ & $(1,2S)$ & $(1,1)$ & $(1,0)$ & $(0,*)$  \\  \thickhline 
$(1,4)$ & $1 - 2/T$ & $4/(3T)$ & 0 & 0 & 0 & 0 & $2/(3T)$  \\
$(1,3)$ & 0 & $1 - 2/T$ & $1/(2T)$ & $2/(3T)$ & 0 & 0 & $5/(6T)$  \\
$(1,2N)$ & 0 & 0 & $1 - 2/T$ & 0 & $2/(3T)$ & 0 & $4/(3T)$  \\
$(1,2S)$ & 0 & 0 & 0 & $1 - 2/T$ & $1/T$ & 0 & $1/T$  \\
$(1,1)$ & 0 & 0 & 0 & 0 & $1 - 2/T$ & $1/(2T)$ & $3/(2T)$  \\
$(1,0)$ & 0 & 0 & 0 & 0 & 0 & $1 - 2/T$ & $2/T$  \\
$(0,*)$ & 0 & 0 & 0 & 0 & 0 & 0 & 1   \\
\end{tabular}
\end{table}

Thus, the resulting Matching Probability per Mass for $\bfi$ is equal to:
\[
\frac{1 - 121/(36 \sfe^2)}{2/3} \approx 0.8177.
\]
Note that in this case, we claim that all of $\bfi$'s offline neighbors of mass $1/3$ have a performance at least as good as in \tbf{Case 2}, since each of them has $\bfi$ as an offline neighbor with mass $2/3$, producing a stronger real-time boosting effect compared with  \tbf{Case 2}.

\begin{lemma}\label{lem:2/3}
Consider an offline node $\bfi$ of mass $2/3$ with two small edges. Its performance when each small edge exhibits a structure in Figure~\ref{fig:wsg/4} is strictly worse than when each exhibits a structure in Figure~\ref{fig:wsg/3}.
\end{lemma}

\begin{proof}
We employ the ODE system to characterize the matching probability of the target offline node $i$. For any time $t \in [0,1]$, let $\alp_t = \E[\saf_{i,t}]$, and let $\eta(t)$ denote the matching rate from each of its online neighbors at time $t$, given that $i$ is safe at that time. By symmetry, the total matching rate from the two online neighbors should be $2\eta(t)$. Thus, the ODE for $i$ is as follows:
\begin{align}\label{eqn:11/26/a}
\alp'(t) &= -\alp(t) \cdot 2\eta(t), \quad \alp(0) = 1.
\end{align}

Let $\eta^a(t)$ and $\eta^b(t)$ be the respective matching rates of $\bfi$ from each online neighbor when the small edges exhibit a structure in Figure~\ref{fig:wsg/3} and Figure~\ref{fig:wsg/4}, respectively. From the analysis in \tbf{Cases 1 and 2}, we see that:
\begin{align*}
\eta^a(t) &= 1 - \frac{2}{3} \frac{1}{1 + \frac{2}{3} t},\\
\eta^b(t) &= 1 - \frac{1}{3 + 2t + \frac{1}{2} t^2} - \frac{2}{3} \frac{1}{1 + \frac{2}{3} t + \frac{1}{6} t^2}.
\end{align*}

We can verify that $\eta^b(t) < \eta^a(t)$ for all $t \in [0,1]$. Thus, by Equation~\eqref{eqn:11/26/a}, we claim that the value of $\alp(t)$ when each small edge exhibits a structure in Figure~\ref{fig:wsg/4} is strictly larger than when they exhibit a structure in Figure~\ref{fig:wsg/3} throughout $t \in [0,1]$.
\end{proof}

\subsection{Modifications for an Offline Node of Mass of 2/3 in Figure~\ref{fig:wsg/3}} \label{app:2/3}

Consider the structure shown in Figure~\ref{fig:wsg/3}. The updated sampling distribution for $j$ is as follows: If both $\bfi$ and $\ti$ are safe upon $j$'s arrival, match $j$ to $\bfi$ with probability $z$ and to $\ti$ with probability $1-z$, where $z \in [0,1]$. If only one is safe, match $j$ to that safe node with probability one. We verify that under the new sampling distribution, $\bfi$ and $\ti$ each achieve an MPM of at least:
\begin{align*}
\kap^{\bfi}(z) &= \frac{\sfe^{-1} z + 1 - 2\sfe^{-1}}{1/3}, \\
\kap^{\ti}(z) &= \frac{\sfe^{-1} (1 - z) + 1 - 2\sfe^{-1}}{2/3}.
\end{align*}

We verify that by setting $z^* = 1 - \frac{\sfe}{3} \approx 0.0939$, we have $\kap^{\bfi}(z^*) = \kap^{\ti}(z^*) = 2 - \frac{3}{\sfe} \approx 0.8963 > \kapb$, where $\kapb$ is the target MPM for any offline node with a mass of 2/3.

\begin{figure}[ht!]
\begin{tikzpicture}
 \draw (0,0) node[minimum size=0.2mm,draw,circle, ultra thick] {$\mathbf{i}$};
  \draw (-0.4,0) node[left] {\bluee{$1/3$}};
    \draw (-0.4,1) node[left] {\bluee{$2/3$}};
 \draw (3,0) node[minimum size=0.2mm,draw,circle, thick] {$j$};
 \draw (0, 1) node[minimum size=0.2mm,draw,circle, thick] {$\ti$};

\draw[-,  very thick] (0.4,0)--(2.6,0);

\draw[-,  red, ultra thick] (0.4, 1)--(2.6,0);
       \draw (1.5,1) node[below] {\bluee{$1-z$}};
              \draw (1.5,0) node[below] {\bluee{$z$}};
\end{tikzpicture}
\caption{Updates proposed for the sampling distribution of online node $j$ in Figure~\ref{fig:wsg/3}. The updated sampling distribution for $j$ is as follows: If both $\bfi$ and $\ti$ are safe upon $j$'s arrival, match $j$ to $\bfi$ and $\ti$ with probabilities $z$ and $1-z$, respectively, where $z \in [0,1]$. If only one is safe, match $j$ to that safe node with probability one.}
    \label{fig:1/7/a}
   \end{figure}
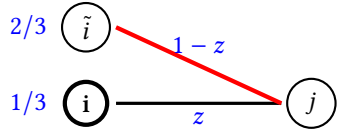

\end{document}